\documentclass[11pt]{amsart}

\usepackage[centertags]{amsmath}
\usepackage{amsthm,amsfonts}

\usepackage{amssymb}
\usepackage{epsfig}
\usepackage{amssymb,latexsym}

\begin{document}
\theoremstyle{plain}
\newtheorem{theorem}{Theorem}[section]
\newtheorem{lemma}[theorem]{Lemma}
\newtheorem{corollary}[theorem]{Corollary}
\newtheorem{proposition}[theorem]{Proposition}

\newtheorem{question}[theorem]{Question}
\theoremstyle{definition}
\newtheorem{notations}[theorem]{Notations}
\newtheorem{notation}[theorem]{Notation}
\newtheorem{remark}[theorem]{Remark}
\newtheorem{remarks}[theorem]{Remarks}
\newtheorem{definition}[theorem]{Definition}
\newtheorem{claim}[theorem]{Claim}
\newtheorem{assumption}[theorem]{Assumption}
\numberwithin{equation}{section}
\newtheorem{examplerm}[theorem]{Example}
\newtheorem{propositionrm}[theorem]{Proposition}

\newtheorem{example}[theorem]{Example}
\newtheorem{examples}[theorem]{Examples}

\newcommand{\binomial}[2]{\left(\begin{array}{c}#1\\#2\end{array}\right)}
\newcommand{\zar}{{\rm zar}}
\newcommand{\an}{{\rm an}}
\newcommand{\red}{{\rm red}}
\newcommand{\codim}{{\rm codim}}
\newcommand{\rank}{{\rm rank}}
\newcommand{\Pic}{{\rm Pic}}
\newcommand{\Div}{{\rm Div}}
\newcommand{\Hom}{{\rm Hom}}
\newcommand{\im}{{\rm im}}
\newcommand{\Spec}{{\rm Spec}}
\newcommand{\sing}{{\rm sing}}
\newcommand{\reg}{{\rm reg}}
\newcommand{\Char}{{\rm char}}
\newcommand{\Tr}{{\rm Tr}}
\newcommand{\tr}{{\rm tr}}
\newcommand{\supp}{{\rm supp}}
\newcommand{\Gal}{{\rm Gal}}
\newcommand{\Min}{{\rm Min \ }}
\newcommand{\Max}{{\rm Max \ }}
\newcommand{\Span}{{\rm Span  }}

\newcommand{\Frob}{{\rm Frob}}
\newcommand{\lcm}{{\rm lcm}}

\newcommand{\soplus}[1]{\stackrel{#1}{\oplus}}
\newcommand{\dlog}{{\rm dlog}\,}    
\newcommand{\limdir}[1]{{\displaystyle{\mathop{\rm
lim}_{\buildrel\longrightarrow\over{#1}}}}\,}
\newcommand{\liminv}[1]{{\displaystyle{\mathop{\rm
lim}_{\buildrel\longleftarrow\over{#1}}}}\,}
\newcommand{\boxtensor}{{\Box\kern-9.03pt\raise1.42pt\hbox{$\times$}}}
\newcommand{\sext}{\mbox{${\mathcal E}xt\,$}}
\newcommand{\shom}{\mbox{${\mathcal H}om\,$}}
\newcommand{\coker}{{\rm coker}\,}
\renewcommand{\iff}{\mbox{ $\Longleftrightarrow$ }}
\newcommand{\onto}{\mbox{$\,\>>>\hspace{-.5cm}\to\hspace{.15cm}$}}

\newenvironment{pf}{\noindent\textbf{Proof.}\quad}{\hfill{$\Box$}}

\newcommand{\sA}{{\mathcal A}}
\newcommand{\sB}{{\mathcal B}}
\newcommand{\sC}{{\mathcal C}}
\newcommand{\sD}{{\mathcal D}}
\newcommand{\sE}{{\mathcal E}}
\newcommand{\sF}{{\mathcal F}}
\newcommand{\sG}{{\mathcal G}}
\newcommand{\sH}{{\mathcal H}}
\newcommand{\sI}{{\mathcal I}}
\newcommand{\sJ}{{\mathcal J}}
\newcommand{\sK}{{\mathcal K}}
\newcommand{\sL}{{\mathcal L}}
\newcommand{\sM}{{\mathcal M}}
\newcommand{\sN}{{\mathcal N}}
\newcommand{\sO}{{\mathcal O}}
\newcommand{\sP}{{\mathcal P}}
\newcommand{\sQ}{{\mathcal Q}}
\newcommand{\sR}{{\mathcal R}}
\newcommand{\sS}{{\mathcal S}}
\newcommand{\sT}{{\mathcal T}}
\newcommand{\sU}{{\mathcal U}}
\newcommand{\sV}{{\mathcal V}}
\newcommand{\sW}{{\mathcal W}}
\newcommand{\sX}{{\mathcal X}}
\newcommand{\sY}{{\mathcal Y}}
\newcommand{\sZ}{{\mathcal Z}}

\newcommand{\A}{{\mathbb A}}
\newcommand{\B}{{\mathbb B}}
\newcommand{\C}{{\mathbb C}}
\newcommand{\D}{{\mathbb D}}
\newcommand{\E}{{\mathbb E}}
\newcommand{\F}{{\mathbb F}}
\newcommand{\G}{{\mathbb G}}
\newcommand{\HH}{{\mathbb H}}
\newcommand{\I}{{\mathbb I}}
\newcommand{\J}{{\mathbb J}}
\newcommand{\M}{{\mathbb M}}
\newcommand{\N}{{\mathbb N}}
\renewcommand{\P}{{\mathbb P}}
\newcommand{\Q}{{\mathbb Q}}
\newcommand{\T}{{\mathbb T}}
\newcommand{\U}{{\mathbb U}}
\newcommand{\V}{{\mathbb V}}
\newcommand{\W}{{\mathbb W}}
\newcommand{\X}{{\mathbb X}}
\newcommand{\Y}{{\mathbb Y}}
\newcommand{\Z}{{\mathbb Z}}

\newcommand{\be}{\begin{eqnarray}}
\newcommand{\ee}{\end{eqnarray}}
\newcommand{\nn}{{\nonumber}}
\newcommand{\dd}{\displaystyle}
\newcommand{\ra}{\rightarrow}
\newcommand{\bigmid}[1][12]{\mathrel{\left| \rule{0pt}{#1pt}\right.}}
\newcommand{\cl}{${\rm \ell}$}
\newcommand{\clp}{${\rm \ell^\prime}$}


\newcommand{\myConstant}{\gamma}

\title[Generalizations on repeated-root constacyclic codes]{Two generalizations on the minimum Hamming distance of repeated-root constacyclic codes }

\author{Hakan \"{O}zadam and Ferruh \"Ozbudak}

\maketitle

\begin{center}
Department of Mathematics and Institute of Applied Mathematics \\
Middle East Technical University, \.{I}n\"on\"u Bulvar{\i}, 06531, Ankara, Turkey \\
\{ozhakan,ozbudak\}@metu.edu.tr
\end{center}

\vspace{0.5cm}

\abstract
	We study constacyclic codes, of length $np^s$\ and $2np^s$, that are generated
	by the polynomials $(x^n + \myConstant)^{\ell}$\ and $(x^n - \xi)^i(x^n + \xi)^j$\ respectively,
	where $x^n + \myConstant$, $x^n - \xi$\ and $x^n + \xi$\ 
	are irreducible over the alphabet $\F_{p^a}$.
	We generalize the results of \cite{D2008},  \cite{OZOZ_1}\ and \cite{OZOZ_2} by computing the minimum
	Hamming distance of these codes.
	As a particular case, we determine the minimum Hamming distance of cyclic and negacyclic codes,
	of length $2p^s$, over a finite field of characteristic $p$.
\endabstract




\vspace{0.5cm}
\section{Introduction}
\label{Section.Introduction}

The minimum Hamming distance of cyclic codes, of length $2^s$, over the Galois ring
$GR(2^a,m)$\ is determined in \cite{D2007}. In \cite{D2008}, the techniques introduced in
\cite{D2007} are used to compute the minimum Hamming distance of cyclic codes,
of length $p^s$, over a finite field of characteristic $p$.

It has been shown, in \cite{CMSS1}, that the minimum Hamming distance of a repeated 
root cyclic code can be expressed in terms of a
simple root cyclic code. 
Using this result in \cite{OZOZ_1}, we have shown that the main result of \cite{D2008}
can be obtained immediately. More explicitly, we have shown that the minimum Hamming distance
of a cyclic code, of length $p^s$, over a finite field of characteristic $p$\
can be found using the results of \cite{CMSS1} via simpler and more direct methods compared
to those of \cite{D2008}.
Later in \cite{OZOZ_2}, we extended our methods, again using the results of \cite{CMSS1}, to cyclic codes,
of length $2p^s$, over a finite field of characteristic $p$, where $p$\ is an odd prime,
and we determined the minimum Hamming distance of these codes.


In this study, we generalize the results of \cite{D2008}, \cite{OZOZ_1} and \cite{OZOZ_2} to certain classes of
repeated-root constacyclic codes. 
Namely, we compute the minimum Hamming distance of constacyclic codes of length $np^s$\ and $2np^s$,
that are generated by the polynomials $(x^n + \lambda )^{\ell}$\ and $(x^n - \xi)^i(x^n + \xi)^j$ respectively,
where $ x^n + \lambda $, $x^n - \xi$\ and $x^n + \xi$\ are irreducible over the alphabet $\F_{p^a}$.
As a particular case, we determine the minimum Hamming distance of cyclic and negacyclic codes,
of length $2p^s$, over a finite field of characteristic $p$.

This paper is organized as follows. In Section \ref{Section.Preliminaries}, 
we give some preliminaries and fix our notation.
In Section \ref{Section.Irreducible}, we determine the minimum Hamming distance of 
constacyclic codes, of length $np^s$, over a finite field of characteristic $p$,
where these code are generated by the irreducible polynomial $x^n + \myConstant$.
In Section \ref{Section.Reducible}, we determine the minimum Hamming distance of
constacyclic codes, of length $2np^s$, over a finite field of characteristic $p$,
that are of the form $\langle (x^n - \xi)^i(x^n + \xi)^j \rangle $\ where
$x^n - \xi$\ and $x^n + \xi$\ are irreducible.
In Section \ref{Section.Examples}, we give several examples as applications of 
the main results of Section \ref{Section.Irreducible} and Section \ref{Section.Reducible}.

\section{Preliminaries}
\label{Section.Preliminaries}
Let $p$\ be a prime number and $\F_q$\ be a finite field of characteristic $p$.
Let $N$\ be a positive integer. Throughout this paper we identify a codeword $c = (c_0,c_1,\dots , c_{N-1})$
over $\F_q$\ with the polynomial $c(x) = c_0 + c_1x + \cdots + c_{N-1}x^{N-1} \in \F_q[x]$.

The \textit{Hamming weight} of a codeword is defined to be the nonzero components of the codeword 
and the \textit{Hamming weight} of a polynomial is defined to be the number of nonzero coefficients
of the polynomial. Let $c$\ and $c(x)$\ be as above. We denote the Hamming weight of
$c$\ and $c(x)$\ by $w_H(c)$\ and $w_H(c(x))$, respectively. 
Obviously, the Hamming weight of a codeword and the Hamming weight of the corresponding
polynomial are equal, i.e., $w_H(c) = w_H(c(x))$.

The \textit{minimum Hamming distance} of a code $C$\ is defined as
\be\nn
	\min \{w_H(u-v):\quad u,v \in C \quad \mbox{and} \quad u \neq v  \},
\ee
and is denoted by $d_H(C)$. If $C$\ is a linear code, then
it is well-known that
\be\nn
	d_H(C) = \min \{w_H(v):\quad 0 \neq v \in C \}.
\ee

Let $\lambda \in \F_q \setminus \{ 0 \}$\ and $I = \langle x^N - \lambda \rangle$.
The $\lambda$-\textit{shift} of a codeword $c = (c_0,c_1,\dots , c_{N-1})$\ is defined
to be $(\lambda c_{N-1}, c_0, c_1,\cdots,c_{N-2})$. If a linear code $C$\ is closed under
$\lambda$-shifts, then $C$\ is called a $\lambda$-cyclic code and in general,
such codes are called \textit{constacyclic} codes (c.f. \cite[Section 13.2]{BRLKMP}). 
It is well-known that
$\lambda$-cyclic codes, of length $N$, over $\F_q$\ correspond 
to the ideals of the finite ring
\be\nn
	\sR = \frac{\F_q[x]}{I}.
\ee
In particular, cyclic (respectively negacyclic) codes, of length $N$,
over $\F_q$ correspond to the ideals of the ring
$\sR_{\mathfrak{a}} = \F_q[x] / \mathfrak{a}$ 
(respectively $\sR_{\mathfrak{b}} = \F_q[x] / \mathfrak{b}$), where
$\mathfrak{a} = \langle x^N -1 \rangle$ (respectively $\mathfrak{b} = \langle x^N + 1 \rangle $).
Additionally if $N$\ is divisible by $p$, then $C$\ is said to be a \textit{repeated-root} constacyclic code.

Any element of $\sR$\ can be represented uniquely as $f(x) + I$\ where
$\deg (f(x)) < N$. The codeword which corresponds to $f(x) + I$\ is
$(f_0,f_1,\dots,f_{N-1})$, where $f(x) = f_0 + f_1x + \cdots + f_{N-1}x^{N-1} \in \F_q[x]$.
Since $\F_q[x]$\ is a principal ideal domain, $\sR$\ is also a principal ideal domain.
So, for any ideal $J$\ of $\sR$, there exists a unique monic polynomial $g(x)\in \F_q[x]$\
with $\deg (g(x)) < N$\ and $g(x) \mid x^N - \lambda $\ such that
$J = \langle g(x) \rangle$. The polynomial $g(x)$\ is said to be a generator of $J$.

The following lemma gives us a trivial lower bound for the minimum Hamming
distance of all constacyclic codes.

\begin{lemma}\label{Preliminaries.Lemma.Weight.Greater.Than.Two}
	Let $\{ 0 \} \neq C \subsetneq \sR$\ be a linear code.
	Then $d_H(C) \ge 2$.
\end{lemma}
\begin{proof}
	Since $x^N \equiv \lambda \mod x^N - \lambda $,
	we have $x^{N(q-1)} \equiv  \lambda^{q-1} \equiv 1 \mod x^N - \lambda $.
	So $x$\ is a unit in $\sR$.
	It is clear from $C \neq \{ 0 \}$\ that $d_H(C) > 0$.
	Now assume that $d_H(C) = 1$. Then there is $\alpha x^e \in C$\ for some
	$\alpha \in \F_q \setminus \{ 0 \} $\ and for some nonnegative integer $e$.
	Since $\alpha$\ and $x$\ are units, $\alpha x^e$\ is a unit in $\sR$.
	Being a proper ideal of $\sR$, $C$\ can not contain a unit.
	Thus we get a contradiction. Hence $d_H(C) \ge 2$.
\end{proof}

Now we will partition the set $\{ 1,2,\dots,p^s-1 \}$\ into three subsets.
These subsets naturally arise from the technicalities of our computations as described in
Section \ref{Section.Irreducible} and Section \ref{Section.Reducible}.
If $i$\ is an integer satisfying $1 \le i \le (p-1)p^{s-1}$, then there exists a uniquely determined
integer $\beta$\ such that $0 \le \beta \le p-2$\ and 
\be\nn
	\beta p^{s-1}+1 \le i \le (\beta + 1)p^{s-1}.
\ee
Moreover since
\be
 	p^s-p^{s-1} < p^s-p^{s-2}< \cdots < p^s-p^{s-s}=p^s-1,
\nn\ee
for an integer $i$ satisfying $(p-1)p^{s-1}+1 = p^s-p^{s-1}+1 \le i \le p^s-1$, there exists a uniquely determined
integer $k$ such that $1 \le k \le s-1$ and
\be \label{condition.on.k}
	p^s-p^{s-k}+1 \le i \le p^s-p^{s-k-1}.
\ee
Besides if $i$ is an integer as above and $k$ is the integer satisfying $1 \le k \le s-1$ and (\ref{condition.on.k}), then we have
\be
    && p^s-p^{s-k} < p^s-p^{s-k}+p^{s-k-1} < p^s-p^{s-k}+2p^{s-k-1} < \cdots \nn \\
    && < p^s-p^{s-k}+(p-1)p^{s-k-1}
\nn\ee
and $p^s-p^{s-k}+(p-1)p^{s-k-1}=p^s-p^{s-k-1}$. So for such integers $i$ and $k$, there exists a uniquely determined integer $\tau$ with $1 \le \tau \le p-1$ such that
\be
    p^s-p^{s-k} + (\tau -1) p^{s-k-1} + 1 \le i \le  p^s - p^{s-k} + \tau p^{s-k-1}.
\nn\ee
Thus
\be\label{Preliminaries.Partition}
\begin{array}{rl}
		&\{1,2,\dots,p^{s-1}\} \sqcup \displaystyle\bigsqcup_{\beta = 1}^{p-2}\{i: \quad \beta p^{s-1} + 1 \le i \le (\beta + 1)p^{s-1} \}\\
		&\sqcup \displaystyle\bigsqcup_{k=1}^{s-1}\displaystyle\bigsqcup_{\tau = 1}^{p - 1}\{i:\quad  p^s-p^{s-k} + (\tau -1) p^{s-k-1} + 1 \le i \le  p^s - p^{s-k} + \tau p^{s-k-1} \}
	\end{array}
\ee
gives us a partition of the set $\{1,2,\dots , p^{s}-1\}$.

Here we fix some notation concerning division and remainders in $\F_q[x]$.
Since $\F_q [x]$\ is a Euclidean domain, for any $f(x)$\ and $ 0\neq g(x) \in \F_q[x]$,
there exist unique polynomials $y(x), r(x) \in \F_q[x] $\ such that
\be\nn
	f(x) = g(x)y(x) + r(x)
\ee
where either $0 \le \deg (r(x)) < \deg (f(x))$\ or $r(x) = 0$.
We define
\be\nn
	f(x) \mod g(x) = r(x),
\ee
and we use the notation $f(x) \equiv r(x) \mod g(x)$\ in the usual sense.

Let $e>0$\ be an integer. For any nonnegative integer $a < p^e$,
there exist uniquely determined integers $0 \le a_0, a_1, \dots ,a_{e-1} \le p-1$\ such that
\be\label{Equality.padic.expansion}
	a = a_{e-1}p^{e-1} + \cdots + a_1p + a_0.
\ee
The expression (\ref{Equality.padic.expansion}) is called the \textit{p-adic expansion}
of $a$.

Let $N$\ be a positive integer and $\myConstant \in \F_q \setminus \{ 0 \}$.
Our computations in Section \ref{Section.Irreducible} and Section \ref{Section.Reducible} are
based on expressing the Hamming weight of an arbitrary nonzero codeword in terms of
$w_H ((x^n + \myConstant)^N)$. In \cite{CMSS2}, the Hamming weight of the polynomial
$(x^n + \myConstant)^N$\ is given as described below.
Let $e,n,N$\ and $0\le b_0,b_1,\dots,b_{e-1} \le p-1$\ 
be positive integers such that $N < p^e$\ and let $\myConstant \in \F_{q} \setminus \{ 0 \}$.
Let
\be\nn
	N = b_{e-1}p^{e-1} + \cdots + b_1p + b_0
\ee
be the p-adic expansion of $N$. Then, by \cite[Lemma 1]{CMSS2}, we have 
\be\label{Equality.weight.x+c.N}
	w_H( (x + \myConstant)^N ) = \prod_{d=0}^{e-1}(b_d + 1).
\ee
As suggested in \cite{CMSS2}, identifying $x$\ with $x^n$ in (\ref{Equality.weight.x+c.N}), we obtain
\be\label{Equality.weight.xn+c.N}
	w_H( (x^n + \myConstant)^N ) = \prod_{d=0}^{e-1}(b_d + 1).
\ee

The following two lemmas are consequences of (\ref{Equality.weight.xn+c.N}) and we will
use them in our computations frequently.

\begin{lemma}\label{Preliminaries.Lemma.weight.beta}
	Let $m,n, 1 \le \beta \le p-2$\ be positive integers and $\myConstant \in \F_q \setminus \{0\}$.
	If $m < p^s - \beta p^{s-1} -1$, then $w_H( (x^n + \myConstant)^{m + \beta p^{s-1} + 1} ) \ge \beta + 2$.
\end{lemma}
\begin{proof}
	Since 
	\be\nn
		m < p^s - \beta p^{s-1} -1 = (p - \beta -1)p^{s-1} + (p-1)p^{s-2} + \cdots + (p-1)p + p-1,
	\ee
	either
	\be\nn
		m & = & Lp^{s-1} + (p-1)p^{s-2} + \cdots + (p-1)p + p-1\quad \mbox{or}\\
		m & = & a_{s-1}p^{s-1} + \cdots + a_1 + a_0\nn
	\ee
	holds, where $0 \le L \le p- \beta - 2$, $0 \le a_0,a_1,\dots,a_{s-2} \le p-1$\ 
	and $0 \le a_{s-1} \le p - \beta -1$\ are integers such that $a_{\ell} < p-1$\ for some $0 \le \ell < s-1$. 
	According to the p-adic expansion of $m$,
	we consider the following two cases.
	
	First, we assume that $m = Lp^{s-1} + (p-1)p^{s-2} + \cdots + (p-1)p + p-1$.
	Then $m + \beta p^{s-1} + 1 = (L + \beta +1)p^{s-1}$. So using (\ref{Equality.weight.xn+c.N}), we get
	\be\nn
		w_H( (x^n + \myConstant)^{m + \beta p^{s-1} + 1} ) = L + \beta + 2 \ge \beta + 2.
	\ee
	
	Second, we assume that $m = a_{s-1}p^{s-1} + \cdots + a_1p + a_0$. Then the p-adic expansion of
	$m + \beta p^{s-1} + 1$\ is of the form
	\be\nn
		m + \beta p^{s-1} + 1 = b_{s-1}p^{s-1} + \cdots + b_1p + b_0
	\ee
	where $0 \le b_0,b_1,\dots,b_{s-2} \le p-1$\ and 
	\be\label{Preliminaries.Lemma.weight.beta.bs1}
		b_{s-1} = a_{s-1} + \beta.
	\ee
	Let $k$\ be the least nonnegative integer with $a_k < p-1$. Then it follows that
	\be\label{Preliminaries.Lemma.weight.beta.bk}
		0 < b_k \le p-1.
	\ee
	So, using (\ref{Equality.weight.xn+c.N}), (\ref{Preliminaries.Lemma.weight.beta.bs1}) and 
	(\ref{Preliminaries.Lemma.weight.beta.bk}), we get
	\be\nn
		w_H( (x^n + \myConstant)^{m + \beta p^{s-1} + 1} )\ge (\beta + a_{s-1} + 1)(b_k + 1) 
								\ge (\beta + 1)2
								> \beta + 2.
	\ee
\end{proof}

\begin{lemma}\label{Preliminaries.Lemma.weight.tau.k}
	Let $m,n, 1 \le \tau \le p-1,1\le k \le s-1 $\ be positive integers and 
	$\myConstant \in \F_q \setminus \{0\} $. If $m < p^{s-k} - (\tau - 1)p^{s-k-1} - 1$, then
	$w_H( (x^{2n} + \myConstant)^{ m + p^s -p^{s-k} + (\tau -1)p^{s-k-1} + 1} ) \ge (\tau + 1)p^k$.
\end{lemma}
\begin{proof}
	Since 
	\be 
		m & < & p^{s-k} - (\tau - 1)p^{s-k-1} -1 \nn\\
		& = & (p - \tau +1)p^{s - k - 1} - 1\nn\\
		& = & (p - \tau)p^{s-k-1} + (p-1)p^{s-k-2}+ \cdots + (p-1)p + p-1,\nn
	\ee
	either
	\be
		m & = & Lp^{s-k-1} + (p-1)p^{s-k-2} + \cdots + (p-1)p + p-1\quad \mbox{or}\nn\\
		m & = & a_{s-k-1}p^{s-k-1} + \cdots + a_1p + a_0\nn
	\ee
	holds, where $0 \le L \le p - \tau -1$, $0 \le a_0,a_1,\dots,a_{s-k-2}\le p-1$\ and 
	$ 0 \le a_{s-k-1} \le p - \tau$\ are some integers such that $0 \le a_{\ell} < p-1$\
	for some $0 \le \ell < s-k-1$. According to the p-adic expansion of $m$, 
	we consider the following two cases.
	
	First, we assume that $m = Lp^{s-k-1} + (p-1)p^{s-k-2} + \cdots + (p-1)p + p-1$.
	Then the p-adic expansion of $m + p^s - p^{s-k} + (\tau -1)p^{s-k-1} + 1$\ is of the form
	\be\nn
		m + p^s - p^{s-k} + (\tau -1)p^{s-k-1} + 1 = (p-1)p^{s-1} + \cdots + (p-1)p^{s-k} + (L + \tau)p^{s-k}.
	\ee
	So, using (\ref{Equality.weight.xn+c.N}), we get
	\be\nn
		w_H( (x^n + \myConstant)^{m+ p^s - p^{s-k} + (\tau -1)p^{s-k-1}+1} ) \ge (\tau +1)p^k.
	\ee
	
	Second, we assume that $m= a_{s-k-1}p^{s-k-1} + \cdots + a_1p + a_0$.
	Then the p-adic expansion of $m + p^s - p^{s-k} + (\tau -1)p^{s-k-1} + 1$\ is of the form
	\be\nn
		m + p^s - p^{s-k} + (\tau -1)p^{s-k-1} + 1 & = &
			(p-1)p^{s-1} + \cdots + (p-1)p^{s-k}\\
			& & + b_{s-k-1}p^{s-k-1}+ \cdots + b_1p + b_0\nn
	\ee
	where $0 \le b_0, b_1,\dots ,b_{s-k-1}\le p-1 $\ are integers.
	It is easy to see that
	\be\label{Preliminaries.Lemma.weight.tau.k.bsk1}
		b_{s-k-1} = a_{s-k-1}  + \tau -1.
	\ee
	Let $\ell_{0}$\ be the least nonnegative  integer with $0 \le a_{\ell_{0}} < p-1$. Then
	\be\label{Preliminaries.Lemma.weight.tau.k.bell0}
		0 < b_{\ell_0} \le  p-1.
	\ee
	Using (\ref{Preliminaries.Lemma.weight.tau.k.bsk1}), (\ref{Preliminaries.Lemma.weight.tau.k.bell0})
	and (\ref{Equality.weight.xn+c.N}), we get
	\be
		w_H( (x^n + \myConstant)^{m + p^s - p^{s-k} (\tau - 1)p^{s-k-1} + 1 } )
			 & \ge &  p^k(b_{s-k-1} + 1)(b_{\ell_0} +1) \nn \\
			& \ge & 2 \tau p^k\nn\\ 
			& \ge & (\tau + 1)p^k.\nn
	\ee
\end{proof}

In \cite{CMSS2}, the authors have shown that the polynomial $(x^n + \myConstant)^N$\
has the so-called \textit{``weight retaining property''} (see \cite[Theorem 1.1]{CMSS2}).
As a result of this, they gave a lower bound for the Hamming weight of the polynomial
$g(x)(x^n + \myConstant)^N$\ where $g(x)$\ is any element of $\F_q[x]$.
Let $n,N,\myConstant $\ and $g(x)$\ be as above. Then, by  \cite[Theorem 1.3 and Theorem 6.3]{CMSS2}, the Hamming weight of $g(x)(x^n + \myConstant)^N$\
satisfies
\be\label{Inequality.Lower.Bound.g(x).xn+c.N}
	w_H(g(x)(x^n + \myConstant )^N) \ge w_H( g(x)\mod x^n + \myConstant )\cdot w_H((x^n + \myConstant )^N).
\ee

As the last remark of this section, we examine the Hamming weight of the polynomials
$(x^n + \myConstant_1)^{p^s}(x^n + \myConstant_2)^i$\ where $0 < i < p^s$.
Let $0 < i < p^s$\ be an integer and 
$\myConstant_1, \myConstant_2 \in \F_q \setminus\{ 0 \}$. Let
\be\nn
	(x^n + \myConstant_2)^i = a_ix^{ni} + a_{i-1}x^{n(i-1)} + \cdots + a_0\myConstant_2^i
\ee
where $a_0,a_1,\dots ,a_i $\ are the binomial coefficients. 
Note that
\be
	(x^n + \myConstant_1)^{p^s}(x^n + \myConstant_2)^i
		& = & (x^{np^s} + \myConstant_1^{p^s})(a_ix^{ni} +  a_{i-1}x^{n(i-1)}\myConstant_2 + \cdots + a_0 \myConstant_{2}^{i}) \nn \\
		& = & a_ix^{n(i+p^s)} + a_{i-1}x^{n(i-1+p^s)}\myConstant_2 + \cdots + a_0x^{np^s}\myConstant_2^i  \nn\\
		& & + a_i\myConstant_1^{p^s}x^{ni} + a_{i-1}\myConstant_1^{p^s}x^{n(i-1)} + \cdots + a_0 \myConstant_1^{p^s}\myConstant_2^i. \nn
\ee
Therefore
\be\nn
	w_H( (x^n + \myConstant_1)^{p^s}(x^n + \myConstant_2)^i ) = 2w_H( (x^n + \myConstant_2)^i ).
\ee
\section{Constacyclic codes of length $np^s$}
\label{Section.Irreducible}

Let $n$\ and $s$\ be positive integers.
Let $\myConstant, \lambda\in \F_q \setminus\{ 0 \}$\ such that $\myConstant^{p^s} = - \lambda$.
All $\lambda$-cyclic codes, of length $np^s$, over $\F_q$ correspond
to the ideals of the finite ring
\be\nn
	\sR = \frac{\F_q[x]}{\langle x^{np^s} - \lambda \rangle}.
\ee 
Suppose that $x^n + \myConstant$\ is irreducible over $\F_q$. 
Then the monic divisors of $x^{np^s} - \lambda = (x^n + \myConstant)^{p^s}$\  
are exactly the elements of the set
$
\{ (x^n + \lambda)^i: \quad 0 \le i \le p^s \}
$. 
So if $x^n + \lambda$\ is irreducible over $\F_q$,
then the $\lambda$-cyclic codes, of length $np^s$,
over $\F_q$\ are of the form $\langle (x^n + \lambda)^i\rangle$\
where $0  \le i \le p^s$.

Let $C = \langle (x^n + \myConstant)^i \rangle$\ where $0 \le i \le p^s$\ is an integer
and $x + \myConstant \in \F_{q}[x]$\ is irreducible. 
Obviously if $i=0$, then $C = \sR$, i.e., $C$\ is the whole space $\F_q^{np^s}$, 
and if $i = p^s$, then $C = \{ 0 \}$. For the remaining values of $i$,
we consider the partition of the set $\{1,2,\dots,p^s-1\}$ given in (\ref{Preliminaries.Partition}). 

If $0 < i \le p^{s-1}$, then $d_H(C)$\ is 2 as shown in Lemma \ref{Irreducible.Lemma.Weight.Exactly.Two}.

For $p^{s-1} < i < p^s$, we first find a lower bound on the Hamming weight of an arbitrary 
nonzero codeword of $C$\ in Lemma \ref{Irreducible.Lemma.Weight.Greater.Than.beta.plus.two} and
Lemma \ref{Irreducible.Lemma.Weight.Greater.Than.tau.k}. 
Next in Corollary \ref{Irreducible.Corollary.Weight.Exactly.beta.plus.two} 
and Corollary \ref{Irreducible.Corollary.Weight.Greater.Than.tau.k}, we show that
there exist codewords in $C$, achieving these previously found lower bounds.
This gives us the minimum Hamming distance of $C$. 
We summarize our results in Theorem \ref{Irreducible.Theorem.Main}.
We close this section by showing that Theorem \ref{Irreducible.Theorem.Main}
gives the minimum Hamming distance of negacyclic codes, of length $2p^s$, over 
$\F_{p^a}$\ where $p \equiv 3 \mod 4$\ and $a$\ is an odd number.

\begin{lemma}\label{Irreducible.Lemma.Weight.Exactly.Two}
	Let $1 \le i \le p^{s-1}$\ be an integer and let $C = \langle (x^n + \myConstant)^i \rangle$.
	Then $d_H(C) = 2$.
\end{lemma}
\begin{proof}
	The claim follows from Lemma \ref{Preliminaries.Lemma.Weight.Greater.Than.Two} 
	and the fact that
	$$(x^n + \myConstant)^{p^{s-1}-i}  (x^n + \myConstant)^i = (x^n + \myConstant)^{p^{s-1}} = 
		x^{np^{s-1}} + \myConstant ^{p^{s-1}} \in C.$$
\end{proof}

Let $C = \langle (x^n + \myConstant)^i\rangle $ for some integer $0 < i < p^s$. For any $0 \neq c(x) \in C$,
there exists a $0 \neq f(x) \in \F_{q}[x]$\ such that $c(x) \equiv f(x)(x^n + \myConstant)^i \mod (x^n+ \myConstant)^{p^s}$.
Dividing $f(x)$\ by $(x^n + \myConstant)^{p^s-i}$, we get
\be\nn
	f(x) = q(x)(x^n + \myConstant)^{p^s-i} + r(x)
\ee
where $q(x), r(x) \in \F_{q}[x]$\ and $0 \le \deg (r(x)) < np^s-ni$\ or $r(x) = 0$ .
We observe that
\be\nn
	c(x) & \equiv & f(x)(x^n + \myConstant)^i\\
		& \equiv & (q(x)(x^n + \myConstant)^{p^s-i} + r(x))(x^n + \myConstant)^i \nn \\
		& \equiv & q(x)(x^n + \myConstant )^{p^s} + r(x)(x^n + \myConstant )^i \nn \\
		& \equiv & r(x)(x^n + \myConstant )^i \mod (x^n + \myConstant )^{p^s}.\nn
\ee
Consequently, for any $0\neq c(x) \in C$, there exists $0 \neq r(x) \in \F_q[x]$\ with
$\deg (r(x)) < np^s-ni$\ such that $c(x) = r(x)(x^n + \myConstant)^i$,
where we consider this equality in $\F_q[x]$.
Therefore the Hamming weight of
$c \in C$\ is equal to the nonzero coefficients of $r(x)(x^n + \myConstant )^i \in \F_{q}[x]$,
i.e., $w_H(c) = w_H(r(x)(x^n + \myConstant)^i)$.

In the following lemma, we give a lower bound on $d_H(C)$\ when $p^{s-1} < i $.
\begin{lemma}\label{Irreducible.Lemma.Weight.Greater.Than.beta.plus.two}
	Let $1 \le \beta \le p-2$\ be an integer and let $C = \langle (x + \myConstant)^{\beta p^{s-1} + 1} \rangle $.
	Then $d_H(C) \ge \beta + 2 $.
\end{lemma}
\begin{proof}
	Let $0 \neq c(x) \in C$, then there exists $0 \neq f(x) \in \F_{q}[x]$\ such that 
	\be\nn
		c(x) \equiv f(x)(x^n + \myConstant)^{\beta p ^{s-1} + 1} \mod (x^n + \myConstant)^{p^s}.
	\ee 
	We may assume that
	$\deg (f(x)) < np^s - n \beta p^{s-1} - n = (p - \beta)np^{s-1} -n$.
	We choose $m$\ to be the largest nonnegative integer with $(x^n + \myConstant)^m | f(x)$.
	Clearly $\deg (f(x)) < (p - \beta)np^{s-1} -n$\ implies $m < (p - \beta)p^{s-1} -1$.
	So, by Lemma \ref{Preliminaries.Lemma.weight.beta}, we get
	\be\label{Irreducible.Lemma.Weight.Greater.Than.beta.plus.two.Weight.xn.Myconstant}
		w_H( (x^n + \myConstant )^{m + \beta p^{s-1}+1} )  \ge \beta + 2.
	\ee	
	For $f(x) = g(x)(x^n + \myConstant)^m$, we have
	\be\nn
		g(x)\mod x^n + \myConstant \neq 0	
	\ee
	by our choice of $m$, so
	\be\label{Irreducible.Lemma.Weight.Greater.Than.beta.plus.two.Modulo.Weight.Positive}
		w_H( g(x) \mod (x^n + \myConstant) ) > 0.
	\ee
	Now using
	(\ref{Irreducible.Lemma.Weight.Greater.Than.beta.plus.two.Weight.xn.Myconstant}), 
	(\ref{Irreducible.Lemma.Weight.Greater.Than.beta.plus.two.Modulo.Weight.Positive})
	and
	(\ref{Inequality.Lower.Bound.g(x).xn+c.N}),
	we obtain
	\be\nn
		w_H(c(x)) & = & w_H(g(x)(x^n + \myConstant)^{m + \beta p^{s-1}+1})\\
			& \ge & w_H( g(x)\mod (x^n + \myConstant ) ) w_H( (x^n + \myConstant)^m )\nn \\
			& \ge & \beta + 2 \nn .
	\ee
	This completes the proof.
\end{proof}

Next we show that the lower bound given in Lemma \ref{Irreducible.Lemma.Weight.Greater.Than.beta.plus.two}
is achieved 
when $p^{s-1} < i \le (p-1)p^{s-1}$\
and this gives us the exact value of $d_H(C)$.
\begin{corollary}\label{Irreducible.Corollary.Weight.Exactly.beta.plus.two}
	Let $1 \le \beta \le p-2$, $\beta p^{s-1} + 1 \le i \le (\beta + 1)p^{s-1} $\
	be integers and
	let $C = \langle (x^n + \myConstant)^i \rangle $. 
	Then $d_H(C) = \beta + 2$.
\end{corollary}
\begin{proof}
	Lemma \ref{Irreducible.Lemma.Weight.Greater.Than.beta.plus.two} and
	 $C \subset \langle (x^n + \myConstant)^{\beta p^{s-1} + 1} \rangle$\ implies
	\be\label{Irreducible.Corollary.Weight.Exactly.beta.plus.two.Lower.Bound}
		d_H(C) \ge \beta + 2.
	\ee
	We know, by (\ref{Equality.weight.xn+c.N}), that
	\be\label{Irreducible.Corollary.Weight.Exactly.beta.plus.two.Weight}
		w_H( (x^n + \myConstant)^{(\beta + 1) p^{s-1}} ) = \beta + 2.
	\ee
	Clearly $(x^n + \myConstant)^{(\beta + 1)p^{s-1}} \in C $\ as $(\beta + 1)p^{s-1} \ge i$.
	Thus (\ref{Irreducible.Corollary.Weight.Exactly.beta.plus.two.Weight}) implies
	\be\label{Irreducible.Corollary.Weight.Exactly.beta.plus.two.Upper.Bound}
		d_H(C) \le \beta + 2.
	\ee
	Combining (\ref{Irreducible.Corollary.Weight.Exactly.beta.plus.two.Lower.Bound}) and
	(\ref{Irreducible.Corollary.Weight.Exactly.beta.plus.two.Upper.Bound}), we get
	$d_H(C) = \beta + 2$.
	
\end{proof}

Having covered the range $p^{s-1} < i \le (p-1)p^{s-1}$, now
we give a lower bound on $d_H(C)$ when $(p-1)p^{s-1} < i < p^s $\ in the following lemma.
\begin{lemma}\label{Irreducible.Lemma.Weight.Greater.Than.tau.k}
	Let $1 \le \tau \le p-1$, $1 \le k \le s-1$\ be
	integers and let $C = \langle (x^n + \myConstant)^{ p^s - p^{s-k} + (\tau -1)p^{s-k-1} + 1} \rangle $.
	Then $d_H(C) \ge (\tau + 1)p^k$.
\end{lemma}
\begin{proof}
	Let $0 \neq c(x) \in C$, then there is $0 \neq f(x) \in \F_q[x]$\ such that 
	\be \nn
		c(x) \equiv f(x)(x^n + \myConstant)^{p^s - p^{s-k} + (\tau -1)p^{s-k-1} + 1} \mod (x^n + \myConstant)^{p^s}.
	\ee
	We may assume that 
	\be \label{Irreducible.Lemma.Weight.Greater.Than.tau.k.deg.f}
		\deg (f(x)) < np^{s-k} -n(\tau -1)p^{s-k-1} - n .
	\ee
	Let $m$\ be the largest nonnegative integer with $(x^n + \myConstant)^m | f(x)$.
	Then there exists $g(x) \in \F_q[x]$\ such that $f(x) = g(x)(x^n + \myConstant)^m$.
	Note that
	\be\nn
		p^{s-k} - (\tau -1)p^{s-k-1} - 1 = (p - \tau)p^{s-k-1} + \cdots + (p-1)p + p-1
	\ee
	and
	\be\nn
		p^s - p^{s-k} + (\tau -1)p^{s-k-1} + 1 = (p-1)p^{s-1} + \cdots + (p-1)p^{s-k} + (\tau -1)p^{s-k-1} + 1.
	\ee
	By (\ref{Irreducible.Lemma.Weight.Greater.Than.tau.k.deg.f}), we have $m < p^{s-k} - (\tau -1)p^{s-k-1} -1$.
	So, by Lemma \ref{Preliminaries.Lemma.weight.tau.k}, we get
	\be\label{Irreducible.Lemma.Weight.Greater.Than.tau.k.Weight.x^n.myConstant}
		w_H( (x^n + \myConstant)^{m + p^s - p^{s-k} + (\tau -1)p^{s-k-1} + 1} )
 		\ge p^k(\tau + 1).
	\ee
	The maximality of $m$\ implies $x^n + \myConstant \nmid g(x)$\ and therefore
	$g(x) \mod x^n + \myConstant \neq 0$. So we have
	\be\label{Irreducible.Lemma.Weight.Greater.Than.tau.k.Modulo.Positive}
		w_H( g(x) \mod x^n + \myConstant ) > 0.
	\ee
	Now using (\ref{Inequality.Lower.Bound.g(x).xn+c.N}),
	(\ref{Irreducible.Lemma.Weight.Greater.Than.tau.k.Weight.x^n.myConstant}) and
	(\ref{Irreducible.Lemma.Weight.Greater.Than.tau.k.Modulo.Positive}),
	we obtain
	\be
		w_H( c(x) ) & = & w_H( g(x)(x^n + \myConstant)^{m + p^s - p^{s-k} + (\tau -1)p^{s-k-1} + 1} ) \nn \\
		& \ge & w_H( g(x) \mod x^n + \myConstant )w_H( (x^n + \myConstant)^{p^s - p^{s-k} + (\tau -1)p^{s-k-1} + 1 + m} ) \nn \\
		& \ge & p^k (\tau + 1)\nn. 		  
	\ee
	This completes the proof.
\end{proof}

For $(p-1)p^{s-1} < i < p^s$, we determine $d_H(C)$\ in 
Corollary \ref{Irreducible.Corollary.Weight.Greater.Than.tau.k} where we show the existence
of a codeword that achieves the lower bound given in 
Lemma \ref{Irreducible.Lemma.Weight.Greater.Than.tau.k}.
\begin{corollary}\label{Irreducible.Corollary.Weight.Greater.Than.tau.k}
	Let $1 \le \tau \le p-1$, $1 \le k \le s-1$\ and $i$\ be integers such that
	\be
		p^s-p^{s-k}+(\tau -1)p^{s-k-1}+1 \le i \le p^s-p^{s-k}+\tau p^{s-k-1}.\nn
	\ee
	Let $C = \langle (x^n + \myConstant)^i \rangle$. Then $d_H(C) = (\tau + 1)p^k$.
\end{corollary}
\begin{proof}
	Lemma \ref{Irreducible.Lemma.Weight.Greater.Than.tau.k} and 
	$C \subset \langle (x^n + \myConstant)^{p^s-p^{s-k}+(\tau - 1) p^{s-k-1} + 1} \rangle $\ implies
	\be\label{Irreducible.Corollary.Weight.Greater.Than.tau.k.Lower.Bound}
		d_H(C) \ge (\tau + 1)p^k.
	\ee
	We know, by (\ref{Equality.weight.xn+c.N}), that
	\be\label{Irreducible.Corollary.Weight.Greater.Than.tau.k.Weight}
		w_H( (x^N + \myConstant)^{p^s-p^{s-k}+ \tau p^{s-k-1}} ) = (\tau + 1)p^k.
	\ee
	Clearly $(x^N + \myConstant)^{p^s-p^{s-k}+ \tau p^{s-k-1}} \in C$\ as 
	$p^s-p^{s-k} + \tau p^{s-k-1} \ge i$. 
	Thus (\ref{Irreducible.Corollary.Weight.Greater.Than.tau.k.Weight}) implies
	\be\label{Irreducible.Corollary.Weight.Greater.Than.tau.k.Upper.Bound}
		d_H(C) \le (\tau + 1)p^k.
	\ee
	Combining (\ref{Irreducible.Corollary.Weight.Greater.Than.tau.k.Lower.Bound}) and
	(\ref{Irreducible.Corollary.Weight.Greater.Than.tau.k.Upper.Bound}),
	we get $d_H(C) = (\tau+1)p^k$.
\end{proof}

We summarize our results in the following theorem.

\begin{theorem}\label{Irreducible.Theorem.Main}
	Let $p$\ be a prime number, $\F_q$\ a finite field of characteristic $p$,
	$\myConstant \in \F_{q}\setminus \{ 0 \}$\ and $n$\ be a positive integer.
	Suppose that $x^n + \myConstant \in \F_{q}[x]$\ is irreducible. Then the $\lambda$-cyclic codes
	over $\F_{q}$, of length $np^s$, are of the form $C[i] = \langle (x^n + \myConstant)^i \rangle$,
	where $0 \le i \le p^s$\ and $\lambda = - \myConstant ^{p^{s}}  $. 
	If $i = 0$, then $C$\ is the whole space $\F_q^{2np^s}$\ and if
	$i = p^s$, then $C$\ is the zero space $\{ \mathbf{0} \}$.
	For the remaining values of $i$, the minimum Hamming
	distance of $C[i]$\ is given by
	 $$ 
	\begin{array}{l}
    		\dd d_H(C[i])
    		 = \left\lbrace
    		\begin{array}{ll}
        		2, &  \mbox{if}\ \ 1 \le i \le p^{s-1},\\
        		\beta+2, & \mbox{if}\ \ \beta p ^{s-1} +1 \le i \le (\beta+1)p^{s-1}\ \mbox{where}\ \ 1 \le \beta \le p-2, \\
        		(\tau + 1)p^k, & \mbox{if}\ \ p^s-p^{s-k}+(\tau -1)p^{s-k-1}+1 \le i \le p^s-p^{s-k}+\tau p^{s-k-1}\\
                	& \mbox{where}\ \ 1 \leq \tau \le p-1\ \ \mbox{and}\ \ 1 \le k \le s-1.
    		\end{array} \right.
 	\end{array}
 	$$
\end{theorem}

\begin{remark}
	If we replace $n$\ with $1$ and $\myConstant$\ with $-1$\ in Theorem \ref{Irreducible.Theorem.Main},
	then we obtain the main results of \cite{D2008} and \cite{OZOZ_1}.
	Namely, we obtain \cite[Theorem 4.11]{D2008} and \cite[Theorem 3.4]{OZOZ_1}.
\end{remark}

In the rest of this section, we assume that $p$\ is an odd prime number and 
$a$\ is a positive integer.

Now we will apply Theorem \ref{Irreducible.Theorem.Main} to a particular case. Namely, we will consider the negacyclic codes
over $\F_{p^a}$ of length $2p^s$. In order to apply Theorem \ref{Irreducible.Theorem.Main},
the polynomial $x^2  + 1$\ must be irreducible over $\F_{p^a}$. A complete irreducibility criterion 
for $x^2+1$\ is given in the following lemma.

\begin{lemma}\label{Irreducible.Lemma.Irreducibility.Criterion}
	Let $p$\ be an odd prime and $a$\ be a positive integer.
	The polynomial $x^2 + 1 \in \F_{p^a}[x]$\ is irreducible
	if and only if $p = 4k + 3$\ for some $k \in \N $\ and
	$a$\ is odd.
\end{lemma}
\begin{proof}
	If $p = 4k+3$\ and $a$\ is odd, then $p^a = 4k_1 + 3$\ for some $k_1 \in \N$.
	So $\F_{p^a} \setminus \{ 0 \}$\ has $4k_1 + 2$\ elements. Now assume that
	there exists $\omega \in \F_{p^a} \setminus \{ 0 \} $\ with $\omega ^ 2 = -1$.
	Then $\omega^4 = 1$. This implies that the multiplicative group of $\F_{p^a}$\ has an element 
	of order $4$. This is a contradiction. Therefore $x^2 + 1$\ has no
	root in $\F_{p^a}$\ and hence $x^2 + 1$\ is irreducible.
	For the converse, if $p = 4k+1$, then $p^a = 4k_2 + 1$\ for some
	$k_2 \in \N$. So $\F_{p^a} \setminus \{ 0 \}$\ has $4k_2$\ elements
	and therefore there exists $\alpha \in \F_{p^a} \setminus \{ 0 \}$\ such that
	$\alpha^{2k_2} = -1$. Having a root in $\F_{p^a}$, $x^2+1$\ is reducible over $\F_{p^a}$.
	If $p= 4k + 3$\ and $a$\ is even, then again we have $p^a = 4k_3 + 1$\ for some
	$k_3 \in \N$\ and similarly $x^2+1$\ is reducible over $\F_{p^a}$. 
	This completes the proof.
\end{proof}

Let $C$\ be a negacyclic code of length $2p^s$\ over $\F_{p^a}$.
If $x^2 + 1$ is irreducible over $\F_{p^a}$,
then the minimum Hamming distance of $C$\ is given in the following theorem.

\begin{theorem}\label{Irreducible.Theorem.Particular}
	Let $p = 4k + 3$\ be a prime for some $k \in \N$ and let $a\in \N$\ be an odd number.
	Then the negacyclic codes
	over $\F_{p^a}$, of length $2p^s$, are of the form $C[i] = \langle (x^2 + 1)^i \rangle$,
	where $0 \le i \le p^s$, and
	 $$ 
	\begin{array}{l}
    		\dd d_H(C[i])
    		 = \left\lbrace
    		\begin{array}{ll}
        		2, &  \mbox{if}\ \ 1 \le i \le p^{s-1},\\
        		\beta+2, & \mbox{if}\ \ \beta p ^{s-1} +1 \le i \le (\beta+1)p^{s-1}\ \mbox{where}\ \ 1 \le \beta \le p-2, \\
        		(\tau + 1)p^k, & \mbox{if}\ \ p^s-p^{s-k}+(\tau -1)p^{s-k-1}+1 \le i \le p^s-p^{s-k}+\tau p^{s-k-1}\\
                	& \mbox{where}\ \ 1 \leq \tau \le p-1\ \ \mbox{and}\ \ 1 \le k \le s-1.
    		\end{array} \right.
 	\end{array}
 	$$
\end{theorem}

For the other values of $p$\ and $a$, $x^2 + 1$\
is reducible over $\F_{p^a}$\ and in this case, we compute the minimum
Hamming distance of $C$\ in Section \ref{Section.Reducible}.

\section{Constacyclic codes of length $2np^s$}
\label{Section.Reducible}
We assume that $p$\ is an odd prime number, $n$\ and $s$\ are positive integers,
$\F_q$\ is a finite field of characteristic $p$\ and 
$\lambda, \xi, \psi \in \F_q \setminus \{ 0 \}$ throughout this section.

Suppose that $\psi ^{p^s} = \lambda$\ and $x^{2n} - \psi$\ factors into two irreducible
polynomials over $\F_q$\ as
\be\label{Reducible.Equation.Factorization}
	x^{2n} - \psi = (x^n - \xi)(x^n + \xi).
\ee
In this section, we will find the minimum Hamming distance of 
$\lambda$-cyclic codes, of length $2np^s$, over $\F_q$\ where
(\ref{Reducible.Equation.Factorization}) is satisfied.
As mentioned before, such $\lambda$-cyclic codes correspond to the ideals
of the finite ring
\be
	\sR = \frac{\F_q[x]}{ \langle x^{2np^s} - \lambda \rangle }.\nn
\ee
Since the monic polynomials dividing $x^{2np^s} - \lambda$\ are exactly
the elements of the set $\{ (x^n - \xi)^i(x^n + \xi)^j: \quad 0 \le i,j \le p^s \}$,
the $\lambda$-cyclic codes, of length $2np^s$, over $\F_q$\ are of the form
$\langle (x^n - \xi)^i(x^n + \xi)^j \rangle$, where $0 \le i,j \le p^s$\ are integers.

Let $C = \langle (x^n - \xi)^i(x^n + \xi)^j \rangle$.
If $(i,j) = (0,0)$, then $C = \sR$. If $(i,j) = (p^s, p^s)$, then
$C = \{ 0 \}$. For the remaining values of $(i,j)$, we consider the partition of the set
$\{ 1,2,\dots,p^s-1 \}$\ given in (\ref{Preliminaries.Partition}).

In order to simplify and improve the presentation of our results,
from Lemma \ref{Reducible.Lemma.Weight.3} till Corollary \ref{Reducible.Corollary.ips.j.tau.K},
we consider only the cases where $i \ge j$ explicitly.
We do so because the cases where $j > i$\ can be treated similarly as the
corresponding case of $i> j$.

Now we give an overview of the results in this section.
If $i = 0$, or $j= 0$, or $0 \le i,j \le p^{s-1}$, then the minimum Hamming distance of 
$C$\ can easily found to be 2 as shown in Lemma \ref{Reducible.Lemma.Weight.i.is.0.or.j.is.0} 
and Lemma \ref{Reducible.Lemma.Weight.i.and.j.small}.

If $0 < j \le p^{s-1}$\ and $p^{s-1} + 1 \le i \le p^s$, then $d_H(C)$\ is computed in
Lemma \ref{Reducible.Lemma.Weight.3}, Corollary \ref{Reducible.Corollary.Weight.3},
Lemma \ref{Reducible.Lemma.Weight.4}\ and Corollary \ref{Reducible.Corollary.Weight.4}.

If $p^{s-1} + 1 \le j \le i \le (p-1)p^{s-1}$, then $d_H(C)$\ is computed in
Lemma \ref{Reducible.Lemma.beta.beta.prime} and Corollary \ref{Reducible.Corollary.beta.beta.prime}.

If $p^{s-1} + 1 \le j \le (p-1)p^{s-1} < i \le p^s -1$, then $d_H(C)$\ is computed
in Lemma \ref{Reducible.Lemma.beta.tauK} and Corollary \ref{Reducible.Corollary.beta.tauK}.

If $(p-1)p^{s-1} + 1 \le j \le i \le p^s -1$, then $d_H(C)$\ is computed in
Lemma \ref{Reducible.Lemma.tauK.same.k}, Corollary \ref{Reducible.Corollary.tauK.same.k},
Lemma \ref{Reducible.Lemma.tauK.different.k} and Corollary \ref{Reducible.Corollary.tauK.different.k}.

Finally if $i=p^s$\ and $0 < j < p^s-1$, then $d_H(C)$\ is computed from 
Lemma \ref{Reducible.Lemma.ips.j.small}
till Corollary \ref{Reducible.Corollary.ips.j.tau.K}.

At the end of this section, we summarize our results in Theorem \ref{Reducible.Theorem.Main}.

We begin our computations with the case where $i = 0$\ or $j=0$.

\begin{lemma}\label{Reducible.Lemma.Weight.i.is.0.or.j.is.0}
	Let $0 < i,j \le p^s$\ be integers, let $C = \langle (x^n - \xi)^i \rangle $\ and
	$D = \langle (x^n + \xi)^j \rangle $. Then $d_H(C) = d_H(D) = 2$.
\end{lemma}
\begin{proof}
	Since 
	\be 
		(x^n - \xi)^{p^s - i}(x^n - \xi)^{i} & = & x^{np^s} - \xi^{p^s} \in C \quad \mbox{and} \nn\\
		(x^n + \xi)^{p^s - j}(x^n + \xi)^{j} & = & x^{np^s} + \xi^{p^s} \in D, \nn
	\ee
	we have $d_H(C), d_H(D) \le 2$. On the other hand, $d_H(C), d_H(D) \ge 2$\ by Lemma
	\ref{Preliminaries.Lemma.Weight.Greater.Than.Two}.
	Hence $d_H(C) = d_H(D) = 2$.
\end{proof}

\begin{lemma}\label{Reducible.Lemma.Weight.i.and.j.small}
	Let $C = \langle (x^n - \xi)^i(x^n + \xi)^j \rangle$, for some integers $0 \le i,j \le p^{s-1}$\
	with $(i,j) \neq (0,0)$. Then $d_H(C) = 2$.
\end{lemma}
\begin{proof}
	By Lemma \ref{Preliminaries.Lemma.Weight.Greater.Than.Two}, we have $d_H(C) \ge 2$
	and
	\be\nn
		 (x^n - \xi)^{i} (x^n + \xi)^{j}(x^n - \xi)^{p^{s-1} - i} (x^n + \xi)^{p^{s-1} - j}
		 = x^{2np^{s-1}} - \xi^{2p^{s-1}} \in C 	
	\ee
	implies that $d_H(2) \le 2$. Hence $d_H(C) = 2$.
\end{proof}

Let $C = \langle(x^n - \xi)^i(x^n + \xi)^j \rangle$\ for some integers $0 \le i,j \le p^s$\ with $(0,0)\neq (i,j) \neq (p^s,p^s) $.
Let $0 \neq c(x) \in C$, then there exists $0 \neq f(x) \in \F_q[x]$\ such that
$c(x) \equiv f(x)(x^n - \xi)^i(x^n + \xi)^j \mod x^{2np^s} - \lambda $.
Dividing $f(x)$\ by $(x^n - \xi)^{p^{s}-i}(x^n + \xi)^{p^s -j}$, we get
\be\nn
	f(x) = q(x)(x^n - \xi)^{p^s - i}(x^n + \xi)^{p^s -j} + r(x)
\ee
where $q(x)$, $r(x) \in \F_{q}[x]$\ and, either $r(x) = 0$\ or $\deg (r(x)) < 2np^s - ni - nj$.
Since
\be
	c(x) & \equiv & f(x)(x^n - \xi)^i(x^n + \xi)^j \nn \\
	& \equiv & (q(x)(x^n - \xi)^{p^s - i}(x + \xi)^{p^s - j} + r(x) )(x^n - \xi)^i(x^n + \xi)^j \nn \\
	& \equiv & q(x)(x^n - \xi)^{p^s}(x + \xi)^{p^s} + r(x)(x^n - \xi)^i(x^n + \xi)^j \nn \\
	& \equiv & r(x)(x^n - \xi)^i(x^n + \xi)^j \mod x^{2np^s} - \lambda , \nn 
\ee
we may assume, without loss of generality, that $\deg (f(x)) < 2np^s - ni - nj$.
Moreover $w_H( r(x)(x^n - \xi)^i(x^n + \xi)^j ) = w_H(c) $\ as 
$\deg (r(x)(x^n - \xi)^i(x^n + \xi)^j) < 2np^s$.

Let $i_0$\ and $j_0$\ be the largest integers with $(x^n - \xi)^{i_0} | f(x)$\ and 
$(x^n + \xi)^{j_0} | f(x)$. Then there exists $g(x) \in  \F_q[x]$\ such that 
$f(x) = (x^n - \xi)^{i_0}(x^n + \xi)^{j_0}g(x)$\ and $(x^n - \xi) \nmid g(x)$, $(x^n + \xi) \nmid g(x)$.
Clearly $\deg (f(x)) < 2np^s - ni -nj $\ implies $i_0 + j_0 < 2p^s - i -j$.
Therefore $i_0 < p^s - i$\ or $j_0 < p^s - j$\ must hold.
 
So if $i_0 \ge p^s - i $, then $j_0 < p^s - j$. For such cases,
the following lemma will be used in our computations.

\begin{lemma}\label{Reducible.Lemma.Weight.Bound.When.i0.large}
	Let $i,j,i_0,j_0$\ be nonnegative integers such that $i \ge j$, $i_0 \ge p^s - i$\ and
	$j_0 < p^s -j$. Let $c(x) = (x^n - \xi)^{i_0 + i} (x^n + \xi)^{j_0 + j} g(x)$\ with
	$x^n - \xi \nmid g(x)$\ and $x^n + \xi \nmid g(x)$.
	Then $w_H( c(x) ) \ge 2w_H( (x^{2n} - \xi ^2)^{j_0 + j} )$.
\end{lemma}
\begin{proof}
	Since $i_0 \ge p^s - i$\ and $-j_0 \ge -p^s + j + 1$, we have $i_0 - j_0 \ge j - i + 1$\ or equivalently $i_0 - j_0 + i - j \ge 1$.
	So
	\be\nn
		c(x) = (x^{2n} - \xi^2)^{j_0 + j}(x^n - \xi)^{i_0 - j_0 + i - j}g(x).
	\ee
	Dividing $(x^n - \xi)^{i_0 - j_0 + i - j}g(x)$\ by $x^{2n}- \xi^2$, we get
	\be\label{Reducible.Lemma.Weight.Bound.When.i0.large.Division}
		(x^n - \xi)^{i_0 - j_0 + i - j}g(x) = (x^{2n} - \xi^2)q(x) + r(x)
	\ee
	for some $q(x), r(x) \in \F_{q}[x]$\ with $r(x) = 0$\ or $\deg (r(x)) < 2n$.
	Let $\theta_1$\ and $\theta_2$\ be any roots of $x^n - \xi$\ and $x^n + \xi$,
	respectively, in some extension of $\F_q$.
	Obviously $\theta_1$\ and $\theta_2$\ are roots of $(x^{2n} - \xi^2)q(x)$.
	First we observe that $r(\theta _1) = 0$\ as $\theta_1$\ is a root of LHS of
	(\ref{Reducible.Lemma.Weight.Bound.When.i0.large.Division}).
	Second we observe that $r(\theta_2) \neq 0$\ as $\theta_2$\ is not a root of
	LHS of (\ref{Reducible.Lemma.Weight.Bound.When.i0.large.Division}).
	So it follows that $r(x)$\ is a nonzero and nonconstant polynomial implying
	$w_H(r(x)) \ge 2$. Therefore
	\be\label{Reducible.Lemma.Weight.Bound.When.i0.large.Bound}
		w_H( (x^n - \xi)^{i_0 - j_0 + i - j}g(x)\mod x^{2n} - \xi^2 )
		= w_H( r(x) ) \ge 2.
	\ee
	Using (\ref{Inequality.Lower.Bound.g(x).xn+c.N}) and
	(\ref{Reducible.Lemma.Weight.Bound.When.i0.large.Bound}), we obtain
	\be\nn
		w_H( c(x) ) & = & w_H( (x^{2n} - \xi^2)^{j_0 + j}(x^n - \xi)^{i_0 - j_0 + i - j}g(x) )\\
			& \ge & w_H( (x^{2n} - \xi^2)^{j_0 + j} )  
				w_H( (x^n - \xi)^{i_0 - j_0 + i - j}g(x)\mod x^{2n} - \xi^2 )\nn \\
			& \ge & 2w_H( (x^{2n} - \xi^2)^{j_0 + j} ).\nn	
	\ee
\end{proof}

Now we have the machinery to obtain the minimum Hamming distance of $C$\
for the ranges $p^{s-1} < i \le p^s$\ and $0 < j\le p^s$.

In what follows, for a particular range of $i$\ and $j$, we first give a lower bound on $d_H(C)$\
in the related lemma. Then in the next corollary, we determine $d_H(C)$\ by showing the
existence of a codeword that achieves the previously found lower bound.

We compute $d_H(C)$\ when $0 < j \le p^{s-1} < i \le 2p^{s-1}$\ in the following lemma and corollary.

\begin{lemma}\label{Reducible.Lemma.Weight.3}
	Let $C = \langle (x^n - \xi)^{p^{s-1} + 1}(x^n + \xi) \rangle$.
	Then $d_H(C) \ge 3$.
\end{lemma}
\begin{proof}
	Pick $0 \neq c(x) \in C$\ where $c(x) \equiv f(x)(x^n - \xi)^{p^{s-1} + 1}(x^n + \xi) \mod x^{2np^s} - \lambda$\
	for some $0 \neq f(x) \in \F_q[x]$\ with $\deg (f(x)) < 2np^s - np^{s-1} - 2n$.
	Let $i_0$\ and $j_0$\ be the largest integers with $(x^n - \xi)^{i_0} | f(x)$\ and
	$ (x^n + \xi)^{j_0} | f(x)$. Then $f(x)$\ is of the form 
	$f(x) = (x^n - \xi)^{i_0}(x^n + \xi)^{j_0}g(x)$\ for some $g(x) \in \F_q[x]$\ with 
	$ x^n - \xi \nmid g(x)$\ and $ x^n + \xi \nmid g(x)$. 
	Note that $i_0 < p^s - p^{s-1} -1$\ or $j_0 < p^s - 1$\ holds.
	
	If $i_0 < p^s -p^{s-1} -1$, then, by Lemma \ref{Preliminaries.Lemma.weight.beta},
	\be\label{Reducible.Lemma.Weight.3.Case.1.xn.xi.gt3}
		w_H( (x^n - \xi)^{i_0 + p^{s-1} + 1} ) \ge 3.
	\ee
	Moreover the inequality
	\be\label{Reducible.Lemma.Weight.3.Case1.mod.Ineq}
		w_H( g(x)(x^n + \xi)^{j_0 + 1}\mod x^n - \xi ) > 0
	\ee
	holds since $x^n - \xi \nmid g(x)$.
	Now using (\ref{Inequality.Lower.Bound.g(x).xn+c.N}), (\ref{Reducible.Lemma.Weight.3.Case.1.xn.xi.gt3}) and
	(\ref{Reducible.Lemma.Weight.3.Case1.mod.Ineq}), we obtain
	
	\be\label{Reducible.Lemma.Weight.3.Case.1}
		\begin{array}{rcl}
			w_H(c(x)) & = & w_H( f(x)(x^n - \xi)^{p^{s-1}+1}(x^n + \xi) )\\
				& = & w_H( (x^n - \xi)^{i_0 + p^{s-1}  + 1} (x^n + \xi)^{j_0 + 1}g(x))\\
				& \ge & w_H( (x^n - \xi)^{i_0 + p^{s-1} + 1} )
					w_H( (x^n + \xi)^{j_0 + 1}g(x)\mod x^n - \xi )\\
				& \ge& 3.
		\end{array}
	\ee
	
	If $i_0 \ge p^s - p^{s-1} -1$, then $j_0 < p^s -1$.
	Clearly $w_H( (x^{2n} - \xi ^2)^{j_0 + j} ) \ge 2$.
	So, by Lemma \ref{Reducible.Lemma.Weight.Bound.When.i0.large}, we have
	\be\label{Reducible.Lemma.Weight.3.Case.2}
		w_H( c(x) ) \ge 2 w_H( (x^{2n} - \xi ^2)^{j_0 + j} ) \ge 4.
	\ee
	
	Now combining (\ref{Reducible.Lemma.Weight.3.Case.1}) and
	(\ref{Reducible.Lemma.Weight.3.Case.2}), we obtain
	$w_H(c(x)) \ge 3$, and hence $d_H(C) \ge 3$.
\end{proof}
\begin{corollary}\label{Reducible.Corollary.Weight.3}
	Let $i,j$\ be integers with $2p^{s-1} \ge i > p^{s-1} \ge j > 0$\ and
	let $C = \langle (x^n - \xi)^{i}(x^n + \xi)^{j} \rangle$. Then $d_H(C) = 3$.
\end{corollary}
\begin{proof}
	Since $C \subset \langle (x^n - \xi)^{p^{s-1} + 1}(x^n + \xi) \rangle$, we know,
	by Lemma \ref{Reducible.Lemma.Weight.3}, that $d_H(C) \ge 3$.
	For $(x^n - \xi)^{p^{s-1}}(x^n + \xi)^{p^{s-1}} \in C$, we have
	\be\nn 
		(x^n - \xi)^{2p^{s-1}}(x^n + \xi)^{2p^{s-1}} = (x^{2n} - \xi^2)^{2p^{s-1}}
	 	= x^{4np^{s-1}} - 2\xi^{2p^{s-1}} x^{2np^{s-1}} + \xi^{4p^{s-1}} .
	\ee
	So $d_H(C) \le 3$\ and hence $d_H(C) = 3$.
\end{proof}

For $2p^{s-1} < i < p^s$\ and $0 < j \le p^{s-1}$, $d_H(C)$\ is computed in the following lemma
and corollary.
\begin{lemma}\label{Reducible.Lemma.Weight.4}
	Let $C = \langle (x^n - \xi)^{2p^{s-1} + 1}(x^n + \xi) \rangle $.
	Then $d_H(C) \ge 4$.
\end{lemma}
\begin{proof}
	Pick $0 \neq c(x) \in C$\ where $c(x) \equiv f(x)(x^n - \xi)^{p^{s-1} + 1}(x^n + \xi) \mod x^{2np^s} - \lambda$\
	for some $0 \neq f(x) \in \F_q[x]$\ with $\deg (f(x)) < 2np^s - np^{s-1} - 2n$.
	Let $i_0$\ and $j_0$\ be the largest integers with $(x^n - \xi)^{i_0} | f(x)$\ and 
	$(x^n + \xi)^{j_0} | f(x)$. Then $f(x)$\ is of the form
	$f(x) = (x^n - \xi)^{i_0}(x^n + \xi)^{j_0}g(x)$\ for some $g(x) \in \F_q[x]$\ with $x^n - \xi \nmid g(x)$\
	and $x^n + \xi \nmid g(x)$. Note that $i_0 < p^s - 2p^{s-1} -1$\ or $j_0 < p^s - 1$\ holds since
	$\deg (f(x)) < 2np^s - 2np^{s-1} - 2n$.
	
	If $i_0 < p^s - 2p^{s-1} - 1$, then, by Lemma \ref{Preliminaries.Lemma.weight.beta}, we have
	\be\label{Reducible.Lemma.Weight.4.Ineq1}
		w_H( (x^n - \xi)^{i_0 + 2p^{s-1} + 1} ) \ge 4.
	\ee
	Since $x^n - \xi \nmid g(x)$,
	\be\label{Reducible.Lemma.Weight.4.Ineq2}
		w_H( g(x)(x^n + \xi)^{j_0 + 1} \mod x^n -\xi ) > 0
	\ee
	holds. Now using (\ref{Reducible.Lemma.Weight.4.Ineq1}), (\ref{Reducible.Lemma.Weight.4.Ineq2}) and
	(\ref{Inequality.Lower.Bound.g(x).xn+c.N}), we obtain
	\be
		w_H( c(x) ) & = & w_H( f(x)(x^n - \xi)^{2p^{s-1}+1}(x^n + \xi) )\nn\\
			& = & w_H( (x^n - \xi)^{i_0 + 2p^{s-1} + 1}(x^n + \xi)^{j_0 + 1}g(x) )\nn\\
			& \ge & w_H( (x^n + \xi)^{j_0 + 1}g(x) \mod x^n - \xi )w_H( (x^n - \xi)^{i_0 + 2p^{s-1} +1} ) \nn\\
			& \ge & 4.\nn
	\ee
	
	If $i_0 \ge p^s - 2p^{s-1} -1$, then $j_0 < p^s -1$.
	Clearly $w_H( (x^{2n}- \xi^2)^{j_0 + 1} )\ge 2$.
	So, by Lemma \ref{Reducible.Lemma.Weight.Bound.When.i0.large}, we have 
	\be\nn
		w_H( c(x) ) \ge 2w_H( (x^{2n} - \xi^2)^{j_0 + 1} )\ge 4.
	\ee
	Hence $d_H(C) \ge 4$.
\end{proof}

\begin{corollary}\label{Reducible.Corollary.Weight.4}
	Let $2p^{s-1} < i < p^s$\ and $0 < j \le p^{s-1}$\ be integers, and let
	$C = \langle (x^n - \xi)^i(x^n + \xi)^j \rangle$. Then $d_H(C) = 4$.
\end{corollary}
\begin{proof}
	Since $C \subset \langle (x^n - \xi)^{2p^{s-1} + 1}(x^n + \xi) \rangle $,
	we know, by Lemma \ref{Reducible.Lemma.Weight.4}, that $d_H(C) \ge 4$.
	For $(x^n - \xi)^{p^s}(x^n + \xi)\in C$, we have
	$w_H( (x^n - \xi)^{p^s}(x^n + \xi) ) = 4$. Thus $d_H(C) \le 4$\ and hence $d_H(C) = 4$.
\end{proof}

Next we consider the cases where $p^{s-1} < j \le i \le p^s$. We begin with computing
$d_H(C)$\ when $p^{s-1} < j \le i \le (p-1)p^{s-1}$\ in the following lemma and corollary.

\begin{lemma}\label{Reducible.Lemma.beta.beta.prime}
	Let $1 \le \beta^{'} \le \beta \le p-2$\ be integers and
	$C = \langle (x^n - \xi)^{\beta p^{s-1} + 1}(x^n + \xi)^{\beta^{'}p^{s-1} + 1}  \rangle$.
	Then $d_H(C) \ge \min \{ \beta + 2, 2(\beta^{'} + 2) \}$.
\end{lemma}
\begin{proof}
	Let $0 \neq c(x) \in C$. Then there exists $0 \neq f(x) \in \F_q[x]$\ such that
	$c(x) \equiv (x^n - \xi)^{\beta p^{s-1} + 1}(x^n + \xi)^{\beta^{'} p^{s-1} + 1} \mod x^{2np^s} - \lambda$.
	We may assume that $\deg (f(x)) < 2n^s - n\beta p^{s-1} - n\beta^{'}p^{s-1} -2n$.
	We consider the cases $\beta = \beta^{'}$\ and $\beta < \beta^{'}$\ separately.
	
	First, we assume that $\beta = \beta^{'}$.
	Then $C = \langle (x^n - \xi)^{\beta p^{s-1} + 1}(x^n + \xi)^{\beta^{'}p^{s-1} + 1} \rangle = 
		\langle (x^{2n} - \xi^2)^{\beta p^{s-1} + 1}\rangle$.
	Let $m$\ be the largest nonnegative integer with $(x^{2n} - \xi^2)^m | f(x)$.
	We have $m < p^s - \beta p^{s-1} - 1$\ as $\deg (f(x)) < 2np^s - 2n\beta p^{s-1}- 2n$.
	So, by Lemma \ref{Preliminaries.Lemma.weight.beta}, we get
	\be\label{Reducible.Lemma.beta.beta.prime.Weight.beta.beta.prime.equal}
		w_H( (x^{2n} - \xi^2)^{\beta p^{s-1} + 1 + m} ) \ge \beta + 2.
	\ee	
	Clearly $f(x)$\ is of the form $f(x) = (x^{2n} - \xi^2)^m g(x)$\ for some $g(x) \in \F_q[x]$\ 
	where $x^{2n} - \xi^2 \nmid g(x)$. So $ g(x) \mod x^{2n} - \xi^2 \neq 0$\ and therefore
	\be\label{Reducible.Lemma.beta.beta.prime.Weight.Modulo.beta.beta.prime.equal}
		w_H (g(x) \mod x^{2n} - \xi^2) > 0.	
	\ee
	So if $\beta = \beta^{'}$, then using
	(\ref{Reducible.Lemma.beta.beta.prime.Weight.beta.beta.prime.equal}),
	(\ref{Reducible.Lemma.beta.beta.prime.Weight.Modulo.beta.beta.prime.equal}) and
	(\ref{Inequality.Lower.Bound.g(x).xn+c.N}), we get
	\be\nn\label{Reducible.Lemma.beta.beta.prime.Case.beta.beta.prime.equal}
		w_H( c(x) ) & = & w_H( (x^{2n} - \xi^2)^{m + \beta p^{s-1} + 1} g(x) )\\
			& \ge &  w_H (g(x) \mod x^{2n} - \xi^2) w_H( (x^{2n} - \xi^2)^{m + \beta p^{s-1} + 1} )\nn\\
			& \ge & \beta + 2 \nn.
	\ee

	Second, we assume that $\beta^{'} < \beta$.
	For $c(x) \equiv f(x)(x^n - \xi)^{\beta p^{s-1} + 1}(x^n + \xi)^{\beta^{'}p^{s-1} + 1} \mod x^{2np^s} - \lambda$,
	let $i_0$\ and $j_0$\ be the largest integers with $(x^n - \xi)^{i_0} | f(x)$\ and $(x^n + \xi)^{j_0} | f(x)$.
	Since $\deg (f(x)) < 2np^s -n\beta p^{s-1} - n\beta^{'}p^{s-1} - 2n$, we have
	\be\nn
		i_0 + j_0 < 2p^s - \beta p^{s-1} - \beta^{'}p^{s-1} -2. 
	\ee
	Thus $i_0 < p^s - \beta p^{s-1} - 1$\ or $j_0 < p^s - \beta^{'}p^{s-1} - 1$\ holds.
	
	If $i_0 < p^s - \beta p^{s-1} - 1$, then, by Lemma \ref{Preliminaries.Lemma.weight.beta}, we have
	\be\label{Reducible.Lemma.beta.beta.prime.weight.i0.Small}
		w_H( (x^n - \xi)^{i_0 + \beta p^{s-1} + 1} ) \ge \beta + 2.
	\ee	
	Note that $(x^n + \xi)^{\beta^{'}p^{s-1} + 1}g(x) \mod x^n - \xi \neq 0$\ since 
	$x^n - \xi \nmid (x^n + \xi)^{\beta^{'}p^{s-1} + 1}g(x)$. Therefore 
	\be\label{Reducible.Lemma.beta.beta.prime.Weight.Modulo.beta.large}
		w_H ((x^n + \xi)^{\beta^{'}p^{s-1} + 1}g(x) \mod x^n - \xi ) > 0.	
	\ee
	Using (\ref{Inequality.Lower.Bound.g(x).xn+c.N}), 
	(\ref{Reducible.Lemma.beta.beta.prime.weight.i0.Small}) and
	(\ref{Reducible.Lemma.beta.beta.prime.Weight.Modulo.beta.large}), we obtain
	\be\label{Reducible.Lemma.beta.beta.prime.Case2.1}
		\begin{array}{rcl}
			w_H( c(x) ) & = & w_H( (x^n - \xi)^{i_0 + \beta p^{s-1} + 1}(x^n + \xi)^{j_0 + \beta^{'} p^{s-1} + 1}g(x) )\\
			& \ge & w_H( (x^n + \xi)^{\beta^{'}p^{s-1} + 1}g(x) \mod x^n - \xi )w_H( (x^n - \xi)^{i_0 + \beta p^{s-1} + 1} )\\
			& \ge &  \beta + 2. 
		\end{array}
	\ee
	
	If $i_0 \ge p^s - \beta p^{s-1} - 1$, then $j_0 < p^s - \beta^{'}p^{s-1} - 1$.
	By Lemma \ref{Reducible.Lemma.Weight.Bound.When.i0.large}, we get
	\be\label{Reducible.Lemma.beta.beta.prime.weight.PreBound.i0.Large}
		w_H( c(x) )\ge 2w_H( (x^{2n} - \xi^2)^{j_0 + \beta^{'}p^{s-1} + 1} ).
	\ee
	For $w_H( (x^{2n} - \xi^2)^{j_0 + \beta^{'}p^{s-1} + 1} )$, 
	we use Lemma \ref{Preliminaries.Lemma.weight.beta} and get
	\be\label{Reducible.Lemma.beta.beta.prime.weight.beta.Prime}
		w_H( (x^{2n} - \xi^2)^{j_0 + \beta^{'}p^{s-1} + 1} ) \ge \beta^{'} + 2.
	\ee
	Combining (\ref{Reducible.Lemma.beta.beta.prime.weight.PreBound.i0.Large}) and
	(\ref{Reducible.Lemma.beta.beta.prime.weight.beta.Prime}), we obtain
	\be\label{Reducible.Lemma.beta.beta.prime.Weight.i0Large}
		w_H( c(x) ) \ge 2(\beta^{'} + 2).
	\ee
	So if $\beta^{'} < \beta$, then, by (\ref{Reducible.Lemma.beta.beta.prime.Case2.1}) and
	(\ref{Reducible.Lemma.beta.beta.prime.Weight.i0Large}), we get that
	\be\nn
		w_H( c(x) ) \ge \min \{ \beta + 2, 2(\beta^{'} + 2) \}.
	\ee
	In both cases, namely $\beta = \beta^{'}$\ and $\beta^{'} < \beta $,
	we have shown that
	\be\nn
		d_H(C) \ge \min \{ \beta + 2, 2(\beta^{'} + 2) \}.
	\ee
\end{proof}

\begin{corollary}\label{Reducible.Corollary.beta.beta.prime}
	Let $j \le i $, $1\le \beta^{'} \le \beta \le p - 2$\ be integers
	such that 
	\be\nn
		\begin{array}{rcccl}
			\beta p^{s-1} + 1 & \le & i & \le & (\beta + 1 )p^{s-1}\quad \mbox{and}\\
			\beta^{'}p^{s-1} + 1 & \le & j & \le & (\beta^{'} + 1 )p^{s-1}.
		\end{array}
	\ee 
	Let $C = \langle (x^n - \xi)^i(x^n + \xi)^j \rangle$.
	Then $d_H(C) = \min \{ \beta + 2, 2(\beta^{'} + 2) \}$.
\end{corollary}
\begin{proof}
	We know, by Lemma \ref{Reducible.Lemma.beta.beta.prime}, 
	that $d_H(C) \ge \min \{ \beta + 2, 2(\beta^{'} + 2) \}$.
	So it suffices to show $d_H(C) \le \min \{ \beta + 2, 2(\beta^{'} + 2) \}$.
	
	First, $(\beta + 1)p^{s-1} \ge i,j$\ implies that
	\be\nn
		(x^n - \xi)^{(\beta + 1)p^{s-1}}(x^n + \xi)^{(\beta + 1)p^{s-1}} = 
		(x^{2n} - \xi^2)^{(\beta + 1)p^{s-1}} \in C.
	\ee
	By (\ref{Equality.weight.xn+c.N}), we get
	\be\nn
		w_H((x^{2n} - \xi^2)^{(\beta + 1)p^{s-1}}) = \beta + 2.
	\ee
	Therefore
	\be\label{Reducible.Corollary.beta.beta.prime.Bound.beta}
		d_H(C) \le \beta + 2.
	\ee
	
	Second, we consider $(x^n - \xi)^{p^s}(x^n + \xi)^{(\beta + 1)p^{s-1}} \in C$.
	Using (\ref{Equality.weight.xn+c.N}) and the fact that $p^s > (\beta + 1)p^{s-1}$,
	we get
	\be\nn
		w_H( (x^n - \xi)^{p^s}(x^n + \xi)^{(\beta + 1)p^{s-1}} )
		= 2w_H( (x^n + \xi)^{(\beta^{'} + 1)p^{s-1}} ) = 2(\beta^{'} + 2).
	\ee
	So
	\be\label{Reducible.Corollary.beta.beta.prime.Bound.beta.prime}
		d_H(C) \le 2(\beta^{'} + 2).
	\ee
	Combining (\ref{Reducible.Corollary.beta.beta.prime.Bound.beta}) and 
	(\ref{Reducible.Corollary.beta.beta.prime.Bound.beta.prime}), 
	we deduce that $d_H(C) = \min \{ \beta + 2, 2(\beta^{'} + 2) \}$.
	Therefore
		$d_H(C) = \min \{ \beta + 2, 2(\beta^{'} + 2) \}$.
\end{proof}

The following lemma and corollary deals with the case where $p^{s-1} < j \le (p-1)p^{s-1} < i < p^s$.

\begin{lemma}\label{Reducible.Lemma.beta.tauK}
	Let $1 \le \tau \le p-1$, $1\le \beta \le p-2$, $1\le k \le s-1$\ be integers and
	$C = \langle (x^n - \xi)^{p^s - p^{s-k} + (\tau - 1)p^{s-k-1} + 1}(x^n + \xi)^{\beta p^{s-1} +1} \rangle$. Then $d_H(C) \ge 2 (\beta + 2)$.
\end{lemma}
\begin{proof}
	Let $0 \neq c(x) \in C$. Then there exists $0 \neq f(x) \in \F_{q}[x]$\
	such that $c(x) \equiv (x^n - \xi)^{p^s - p^{s-k} + (\tau -1)p^{s-k-1} + 1}(x^n + \xi)^{\beta p^{s-1} + 1}f(x) \mod x^{2np^s}- \lambda$\ and $\deg (f(x)) < np^s + np^{s-k} - n(\tau - 1)p^{s-k-1} - n \beta p^{s-1} -2n$. Let $i_0$\ and $j_0$\ be the largest integers with $(x^n - \xi)^{i_0} | f(x)$\ and
	$(x^n + \xi)^{j_0} | f(x)$. Then $f(x)$\ is of the form
	$f(x) = (x^n - \xi)^{i_0}(x^n + \xi)^{j_0}g(x)$\ for some $g(x) \in \F_q[x]$\ such that
	$x^n - \xi \nmid g(x)$\ and $x^n + \xi \nmid g(x)$. 
	Clearly $i_0 + j_0 < p^s + p^{s-k} - (\tau - 1)p^{s-k-1} -  \beta p^{s-1} -2$.
	So $i_0 < p^{s-k} - (\tau -1)p^{s-k-1} - 1$\ or 
	$j_0 < p^s - \beta p^{s-1} -1$\ holds.
	
	If $i_0 < p^{s-k} - (\tau -1)p^{s-k-1} - 1$, then, by Lemma 
	\ref{Reducible.Lemma.Weight.Bound.When.i0.large},
	we have
	\be\label{Reducible.Lemma.beta.tauK.i0.small.Ineq1}
		w_H( (x^n - \xi)^{i_0 + p^s - p^{s-k} + (\tau -1)p^{s-k-1} + 1} ) \ge (\tau + 1)p^k.
	\ee
	Since $x^n - \xi \nmid g(x)$,
	\be\label{Reducible.Lemma.beta.tauK.i0.small.Ineq2}
		w_H( (x^n + \xi)^{j_0 + \beta p^{s-1} + 1}g(x) \mod x^n - \xi) > 0.
	\ee
	Using (\ref{Reducible.Lemma.beta.tauK.i0.small.Ineq1}),
	(\ref{Reducible.Lemma.beta.tauK.i0.small.Ineq2}) and 
	(\ref{Inequality.Lower.Bound.g(x).xn+c.N}), we obtain
	\be\nn
		w_H( c(x) ) & = & w_H( (x^n - \xi)^{i_0 + p^s - p^{s-k} + (\tau -1)p^{s-k-1} + 1}
				(x^n + \xi)^{j_0 + \beta p^{s-1} + 1}g(x) )\nn\\
			& \ge & w_H( (x^n + \xi)^{j_0 + \beta p^{s-1} + 1}g(x)\mod x^n - \xi )
				w_H((x^n - \xi)^{i_0 + p^s - p^{s-k} + (\tau -1)p^{s-k-1} + 1})\nn\\
			& \ge & (\tau + 1)p^k\nn\\
			& \ge & 2p\nn\\
			& \ge & 2(\beta + 2).\nn
	\ee
	
	If $i_0 \ge p^{s-k} - (\tau -1)p^{s-k-1} -1$, then 
	$j_0 < p^s - \beta p^{s-1} -1$. So, by Lemma \ref{Reducible.Lemma.Weight.Bound.When.i0.large},
	we get
	\be\label{Reducible.Lemma.beta.tauK.io.Large.Ineq}
		w_H( c(x) ) \ge 2w_H( (x^2 - \xi^2)^{j_0 + \beta p^{s-1} + 1} ).
	\ee
	For $w_H( (x^2 - \xi^2)^{j_0 + \beta p^{s-1} + 1} )$,  
	we use Lemma \ref{Preliminaries.Lemma.weight.beta} and get
	\be\label{Reducible.Lemma.beta.tauK.io.Large.Eq}
		w_H( (x^{2n} - \xi^2)^{j_0 + \beta p^{s-1} + 1} ) = \beta + 2.
	\ee
	Combining (\ref{Reducible.Lemma.beta.tauK.io.Large.Ineq})\ and
	(\ref{Reducible.Lemma.beta.tauK.io.Large.Eq}), we obtain
	$w_H( c(x) ) \ge \beta + 2$.
	So $d_H(C) \ge \beta + 2$.
\end{proof}

\begin{corollary}\label{Reducible.Corollary.beta.tauK}
	Let $i,j,1 \le \tau \le p-1, 1\le \beta \le p -2$\ and
	$1 \le k \le s-1$\ be integers such that 
	\be\nn
		\begin{array}{rcccl}
			p^s - p^{s-k} + (\tau -1)p^{s-k-1} + 1 & \le & i & \le & p^s - p^{s-k} + \tau p^{s-k-1}\quad \mbox{and}\\
			\beta p^{s-1} + 1 & \le & j & \le & (\beta + 1)p^{s-1}.
		\end{array}
	\ee
	Let $C = \langle (x^n - \xi)^i(x^n + \xi)^j \rangle $. 
	Then $d_H(C) = 2 (\beta + 2)$.
\end{corollary}
\begin{proof}
	Since $\langle (x^n - \xi)p^{p^s - p^{s-k}+ (\tau -1)p^{s-k-1} +1}(x^n + \xi)^{\beta p^{s-1} + 1} \rangle \subset C$, we know, by Lemma \ref{Reducible.Lemma.beta.tauK}, that
	$d_H(C) \ge 2(\beta + 2)$. So it suffices to show $d_H(C) \le 2(\beta + 2)$.
	We consider $(x^n - \xi)^{p^s}(x^n + \xi)^{\beta p^{s-1} + 1} \in C$.
	Note that
	\be\nn
		w_H( (x^n - \xi)^{\beta p^{s-1} + 1} ) = \beta + 2
	\ee
	by Lemma \ref{Equality.weight.xn+c.N}. So, using the fact that
	$p^s > \beta p^{s-1} + 1$, we obtain
	\be\nn 
		w_H( (x^n - \xi)^{p^s}(x^n + \xi)^{\beta p^{s-1} + 1} ) = 2 (\beta + 2).
	\ee
	So $d_H(C) \le 2(\beta + 2)$, and hence $d_H(C) = 2(\beta + 2)$.
\end{proof}

From Lemma \ref{Reducible.Lemma.tauK.same.k} till Corollary \ref{Reducible.Corollary.tauK.different.k},
we compute $d_H(C)$\ when $(p-1)p^{s-1} < j \le i < p^s$. 

\begin{lemma}\label{Reducible.Lemma.tauK.same.k}
	Let $1 \le k \le s-1, 1 \le \tau^{'} \le \tau \le p-1$,
	\be
		i & = & p^s - p^{s-k} + (\tau - 1)p^{s-k-1} + 1 \quad \mbox{and}\nn\\
		j & = & p^s - p^{s-k} + (\tau^{'} - 1)p^{s-k-1} + 1 \nn 
	\ee
	be integers and $C = \langle (x^n - \xi)^i(x^n + \xi)^j \rangle$.
	Then $d_H(C) \ge \min \{ 2(\tau^{'}+1)p^k, (\tau + 1)p^k \}$.
\end{lemma}
\begin{proof}
	Let $0 \neq c(x) \in C$. Then there exists $0 \neq f(x) \in \F_q[x]$\ such that
	$c(x) \equiv f(x)(x^n - \xi)^i(x^n + \xi)^j \mod x^{2np^s} - \lambda$\ and
	$\deg (f(x)) < 2np^s - in -jn$. Let $i_0$\ and $j_0$\ be the largest integers
	with $(x^n - \xi)^{i_0} | f(x)$\ and $(x^n + \xi)^{j_0} | f(x)$.
	Then $f(x)$\ is of the form 
	$f(x) = (x^n - \xi)^{i_0}(x^n + \xi)^{j_0}g(x)$\ for some $g(x) \in \F_q[x]$\
	with $x^n - \xi \nmid g(x)$\ and $x^n + \xi \nmid g(x)$.
	Clearly $i_0 + j_0 < 2p^{s} - i - j$\ and therefore $i_0 < p^s -i$\ or $j_0 < p^s -j$
	holds.
	
	If $i_0 < p^s -i$, then by Lemma \ref{Preliminaries.Lemma.weight.tau.k}, we have
	\be\label{Reducible.Lemma.tauK.same.k.small.io.Ineq1}
		w_H( (x^n - \xi)^{i_0 + i} ) \ge (\tau + 1)p^k.
	\ee
	Since $x^n - \xi \nmid g(x)$, we have
	$g(x)(x^n + \xi)^{j_0 + j} \not \equiv 0 \mod x^n - \xi$\ and therefore
	\be\label{Reducible.Lemma.tauK.same.k.small.io.Ineq2}
		w_H( g(x)(x^n + \xi)^{j + j_0} \mod x^n - \xi ) > 0.
	\ee
	Using (\ref{Reducible.Lemma.tauK.same.k.small.io.Ineq1}),
	(\ref{Reducible.Lemma.tauK.same.k.small.io.Ineq2}) and 
	(\ref{Inequality.Lower.Bound.g(x).xn+c.N}), we obtain
	\be\label{Reducible.Lemma.tauK.same.k.small.io}
		\begin{array}{rcl}
			w_H( c(x) ) & = & w_H( (x^n - \xi)^{i + i_0}(x^n + \xi)^{j + j_0}g(x) )\\
				& \ge & w_H( g(x)(x^n + \xi)^{j + j_0}g(x) \mod x^n - \xi )
					w_H( (x^n - \xi)^{i + i_0} ) \\
				& \ge & (\tau + 1)p^k.
		\end{array}
	\ee
	
	If $i_0 \ge p^s -i$, then $j_0 < p^s - j$. 
	So, by Lemma \ref{Reducible.Lemma.Weight.Bound.When.i0.large},
	we have
	\be\label{Reducible.Lemma.tauK.same.k.large.io.Ineq1}
		w_H( c(x) ) \ge 2w_H( (x^{2n} - \xi^2)^{j_0 + j} ).
	\ee
	For $w_H( (x^{2n} - \xi^2)^{j_0 + j} )$, 
	we use Lemma \ref{Preliminaries.Lemma.weight.tau.k} and get
	\be\label{Reducible.Lemma.tauK.same.k.large.io.Ineq2}
		w_H( (x^{2n} - \xi^2)^{j_0 + j} ) \ge (\tau^{'} + 1)p^{k^{'}}.
	\ee
	Combining (\ref{Reducible.Lemma.tauK.same.k.large.io.Ineq1}) and
	(\ref{Reducible.Lemma.tauK.same.k.large.io.Ineq2}), we obtain
	\be\label{Reducible.Lemma.tauK.same.k.large.io}
		w_H( c(x) ) \ge 2(\tau^{'} + 1)p^{k^{'}}.
	\ee
	Now, using (\ref{Reducible.Lemma.tauK.same.k.small.io}) and
	(\ref{Reducible.Lemma.tauK.same.k.large.io}), we deduce that
	$w_H(c(x)) \ge \min\{ 2(\tau^{'} + 1)p^{k^{'}}, (\tau + 1)p^k \}$.
	Hence $d_H(C) \ge \min \{ 2(\tau^{'}+1)p^k, (\tau + 1)p^k \}$.
\end{proof}

\begin{corollary}\label{Reducible.Corollary.tauK.same.k}
	Let $j \le i $, $1\le k \le s-1$, $1\le \tau^{'} \le \tau \le p-1 $\
	be integers such that
	\be\nn
		\begin{array}{rcccl}
			p^s - p^{s-k} + (\tau -1)p^{s-k-1} + 1 & \le & i & \le & p^s - p^{s-k} + \tau p^{s-k-1} \quad \mbox{and} \\
			p^s - p^{s-k} + (\tau^{'} -1)p^{s-k-1} + 1 & \le & j & \le & p^s - p^{s-k} + \tau^{'} p^{s-k-1}.
		\end{array}
	\ee
	Let $C = \langle (x^n - \xi)^i(x^n + \xi)^j \rangle$. Then
	\be\nn
		d_H( C ) = \min \{ 2(\tau^{'} +1)p^{k^{'}}, (\tau + 1)p^k \}.
	\ee
\end{corollary}
\begin{proof}
	Since $\langle (x^n - \xi)^{p^s - p^{s-k} + (\tau -1)p^{s-k-1} + 1}(x^n + \xi)^{p^s - p^{s-k} + (\tau^{'} -1)p^{s-k-1} + 1} \rangle \subset C$, we have, 
	by Lemma \ref{Reducible.Lemma.tauK.same.k}, that
	$d_H(C) \ge \min \{ 2(\tau^{'} +1)p^{k^{'}}, (\tau + 1)p^k \}$.
	So it suffices to show $d_H(C) \le \min \{ 2(\tau^{'} +1)p^{k^{'}}, (\tau + 1)p^k \}$.
	
	First, we consider $(x^n - \xi)^{p^s}(x^n + \xi)^{p^s - p^{s-k} + (\tau^{'} - 1)p^{s-k-1}-1}\in C$.
	Since
	\be\nn
		w_H( (x^n + \xi)^{p^s - p^{s-1} + (\tau^{'} -1)p^{s-k-1}} ) = (\tau^{'} + 1)p^{k^{'}},
	\ee
	we have
	\be\nn
		w_H( (x^n - \xi)^{p^s}(x^n + \xi)^{p^s - p^{s-1} + (\tau^{'} -1)p^{s-k-1}} ) = 2(\tau ^{'} + 1)p^{k^{'}}.
	\ee
	Thus
	\be\label{Reducible.Corollary.tauK.same.k.Ineq1}
		d_H(C) \le 2(\tau^{'} + 1)p^{k^{'}}
	\ee
	
	Second, we consider $(x^{2n} - \xi^2)^{p^s - p^{s-k} + (\tau -1)p^{s-k-1} + 1} \in C$.
	By Lemma \ref{Equality.weight.xn+c.N}, we get
	\be\nn
		w_H( (x^{2n} - \xi^2)^{p^s - p^{s-k} + (\tau -1)p^{s-k-1} + 1} ) = (\tau + 1)p^k.
	\ee
	Thus
	\be\label{Reducible.Corollary.tauK.same.k.Ineq2}
		d_H(C) \le (\tau + 1)p^k.
	\ee
	Now combining (\ref{Reducible.Corollary.tauK.same.k.Ineq1}) and
	(\ref{Reducible.Corollary.tauK.same.k.Ineq2}), we deduce that
	$d_H(C) \le \min \{ 2(\tau^{'} +1)p^{k^{'}}, (\tau + 1)p^k \}$.
	Hence $d_H(C) = \min \{ 2(\tau^{'} +1)p^{k^{'}}, (\tau + 1)p^k \}$.
\end{proof}

\begin{lemma}\label{Reducible.Lemma.tauK.different.k}
	Let $1 \le k^{'} < k \le s-1$, $1\le \tau^{'},\tau < p-1$,
	\be
		i & = & p^s - p^{s-k} + (\tau -1)p^{s-k-1} + 1 \quad \mbox{and} \nn \\
		j & = & p^s - p^{s-k^{'}} + (\tau^{'} -1)p^{s-k^{'}-1} + 1 \nn
	\ee
	be integers and $C = \langle (x^n - \xi)^i(x^n + \xi)^j \rangle $.
	Then $d_H(C) \ge 2(\tau^{'} + 1)p^{k^{'}}$.
\end{lemma}
\begin{proof}
	Let $0 \neq c(x) \in C$. Then there exists $0 \neq f(x) \in \F_q[x]$\ such that
	$c(x) \equiv (x^n - \xi)^i(x^n + \xi)^jf(x) \mod x^{2np^s} - \lambda$\ and
	$\deg (f(x)) < 2np^s -ni - nj$.
	Let $i_0$\ and $j_0$\ be the largest integers with $(x^n - \xi)^{i_0} | f(x)$\
	and $(x^n + \xi)^{j_0} | f(x)$. Then $f(x)$\ is of the form
	$f(x) = (x^n - \xi)^{i_0}(x^n + \xi)^{j_0}g(x)$\ for some $g(x) \in \F_q[x]$\ with
	$x^n - \xi \nmid g(x)$\ and $x^n + \xi \nmid g(x)$.
	Clearly $i_0 + j_0 < 2p^s - i- j$. So $i_0 < p^s - i$\ or $j_0 < p^s - j$\ holds.
	
	If $i_0 < p^s - i$, then, by Lemma \ref{Preliminaries.Lemma.weight.tau.k}, we have
	\be\label{Reducible.Lemma.tauK.different.k.Small.i0.Ineq1}
		w_H( (x^n - \xi)^{i + i_0} ) \ge (\tau + 1)p^k \ge 2(\tau^{'} + 1)p^{k^{'}}.
	\ee
	Since $x^n - \xi \nmid g(x)$, we have
	$(x^n + \xi)^{j_0 + j}g(x) \mod x^n - \xi \neq 0$\ and therefore
	\be\label{Reducible.Lemma.tauK.different.k.Small.i0.Ineq2}
		w_H( (x^n + \xi)^{j_0 + j}g(x) \mod x^n - \xi ) > 0.
	\ee 
	Using (\ref{Reducible.Lemma.tauK.different.k.Small.i0.Ineq1}),
	(\ref{Reducible.Lemma.tauK.different.k.Small.i0.Ineq2}) and
	(\ref{Inequality.Lower.Bound.g(x).xn+c.N}), we obtain
	\be
		w_H( c(x) ) & = & w_H( (x^n - \xi)^{i_0 + i}(x^n + \xi)^{j_0 + j}g(x) )\nn \\
			& \ge & w_H( (x^n + \xi)^{j_0 + j}g(x) \mod x^n - \xi )
				w_H( (x^n - \xi)^{i_0 + i} )\nn \\
			& \ge & 2(\tau^{'} + 1)p^{k^{'}}.\nn
	\ee
	
	If $i_0 \ge p^s - i$, then $j_0 < p^s - j$.
	So, by Lemma \ref{Reducible.Lemma.Weight.Bound.When.i0.large}, we have
	\be\label{Reducible.Lemma.tauK.different.k.Large.i0.Ineq1}
		w_H( c(x) ) \ge 2w_H( (x^{2n} - \xi^2)^{j_0 + j} ).
	\ee
	For $w_H( (x^{2n} - \xi^2)^{j_0 + j} )$,
	we use Lemma \ref{Preliminaries.Lemma.weight.tau.k} and get
	\be\label{Reducible.Lemma.tauK.different.k.Large.i0.Ineq2}
		w_H( (x^{2n} - \xi^2)^{j_0 + j} ) \ge (\tau^{'} + 1)p^{k^{'}}.
	\ee
	Now combining (\ref{Reducible.Lemma.tauK.different.k.Large.i0.Ineq1})
	and (\ref{Reducible.Lemma.tauK.different.k.Large.i0.Ineq2}), we obtain
	$ w_H( c(x) ) \ge 2(\tau ^ {'} + 1)p^{k^{'}}$.
	Hence $d_H(C) \ge 2 (\tau^{'} + 1)p^{k^{'}}$.
\end{proof}

\begin{corollary}\label{Reducible.Corollary.tauK.different.k}
	Let $i,j,1\le k^{'} < k \le s-1, 1\le \tau^{'},\tau \le p-1$\ be integers
	such that
	\be\nn
		\begin{array}{rcccl}
			 p^s - p^{s-k} + (\tau - 1)p^{s-k-1} + 1 & \le & i & \le & p^s - p^{s-k} + \tau p^{s-k-1} \quad \mbox{and}\\
			 p^s - p^{s-k^{'}} + (\tau^{'} - 1)p^{s-k^{'}-1} + 1&\le & j & \le & p^s - p^{s-k^{'}} + \tau p^{s-k^{'}-1}.
		\end{array}
	\ee
	Let $C = \langle (x^n - \xi)^i (x^n + \xi)^j \rangle$.
	Then $d_H(C) = 2 (\tau^{'} + 1)p^{k^{'}}$.
\end{corollary}
\begin{proof}
	Since $\langle (x^n - \xi)^{p^s - p^{s-k} + (\tau -1)p^{s-k-1}+1} (x^n 4 \xi)^{p^s - p^{s-k^{'}} + (\tau^{'} -1)p^{s-k^{'}-1}+1} \rangle \in C$, we know, by Lemma
	\ref{Reducible.Lemma.tauK.different.k}, that $d_H(C) \ge 2(\tau^{'} + 1)p^{k^{'}}$.
	So it suffices to show $d_H(C) \le 2(\tau^{'} + 1)p^{k^{'}}$.
	We consider $(x^n - \xi)^{p^s}(x^n + \xi)^{p^s - p^{s-k^{'}} +1} \in C$.
	By Lemma \ref{Equality.weight.xn+c.N}, we have
	\be\nn
		w_H( (x^n + \xi)^{p^s - p^{s-k^{'}} \tau p^{s-k^{'}-1}} ) = (\tau ^{'} + 1)p^{k^{'}}.
	\ee
	Moreover since $(x^n - \xi)^{p^s} = x^{np^s} - \xi ^{p^s}$\ and
	$p^s > p^s - p^{s-k^{'}} +1$, we get
	\be\nn
		w_H( (x^n - \xi)^{p^s}(x^n + \xi)^{p^s - p^{s-k^{'}} +1} ) = 2(\tau ^{'} + 1)p^{k^{'}}.
	\ee
	So $d_H(C) \le 2(\tau ^{'} + 1)p^{k^{'}}$\ and therefore $d_H(C) = 2(\tau ^{'} + 1)p^{k^{'}}$.
\end{proof}

Finally it remains to consider the cases where $i = p^s$\ and $0 < j < p^s$.

\begin{lemma}\label{Reducible.Lemma.ips.j.small}
	Let $C = \langle (x^n - \xi)^{p^s}(x^n + \xi) \rangle$.
	Then $d_H(C) \ge 4$.
\end{lemma}
\begin{proof}
	Pick $0 \neq c(x) \in C$. Then there exists $0 \neq f(x) \in \F_q[x]$\
	such that $c(x) \equiv f(x)(x^n - \xi)^{p^s}(x^n + \xi) \mod x^{2np^s} - \lambda$\ and
	$\deg (f(x)) < 2np^s - np^s -n = np^s - n$. Let $i_0$\ and $j_0$\ be the largest
	nonnegative integers such that $(x^n - \xi)^{i_0} | f(x)$\ and $(x^n + \xi)^{j_0} | f(x)$.
	Clearly $i_0 + j_0 < p^s -1$\ as $\deg (f(x)) < np^s - n$.
	So, since $i_0 \ge p^s - p^s = 0$\ and $j_0 < p^s -1$, 
	by Lemma \ref{Reducible.Lemma.Weight.Bound.When.i0.large},
	we get
	\be\label{Reducible.Lemma.ips.j.small.Ineq1}
		w_H( c(x) ) \ge 2 w_H( (x^{2n} - \xi^2)^{j_0 + 1} ).
	\ee
	Obviously $w_H((x^{2n} - \xi^2)^{j_0 + 1}) \ge 2$\ and therefore, 
	by (\ref{Reducible.Lemma.ips.j.small.Ineq1}), we obtain
	$w_H( c(x) ) \ge 4$. Hence $d_H( C ) \ge 4$.
\end{proof}
\begin{corollary}\label{Reducible.Corollary.ips.j.small}
	Let $0 < j \le p^{s-1}$\ be an integer and
	$\langle (x^n - \xi)^{p^s}(x^n + \xi)^j \rangle$.
	Then $d_H(C) = 4$.
\end{corollary}
\begin{proof}
	Since $\langle (x^n - \xi)^{p^s}(x^n + \xi) \rangle \in C$,
	we know, by Lemma \ref{Reducible.Lemma.ips.j.small}, that $d_H(C) \ge 4$.
	So it suffices to show $d_H(C) \le 4$. 
	We consider $(x^n - \xi)^{p^s}(x^n + \xi)^{p^{s-1}} \in C$.
	Clearly $w_H( (x^n - \xi)^{p^s}(x^n + \xi)^{p^{s-1}} ) = 4$.
	So $d_H(C) \le 4$\ and hence $d_H(C) = 4$.
\end{proof}

For $i = p^s$\ and $p^{s-1} < j < p^s$, the minimum Hamming distance of $C$\ is
computed in the following lemmas and corollaries. Their proofs are similar to those of
Lemma \ref{Reducible.Lemma.ips.j.small} and Corollary \ref{Reducible.Lemma.ips.j.small}.

\begin{lemma}\label{Reducible.Lemma.ips.j.beta}
	Let $1 \le \beta \le p-2$\ be an integer and 
	$C = \langle (x^n - \xi)^{p^s}(x^n + \xi)^{\beta p^{s-1} +1} \rangle$.
	Then $d_H(C) \ge 2(\beta + 2)$.
\end{lemma}

\begin{corollary}\label{Reducible.Corollary.ips.j.beta}
	Let $1 \le \beta \le p-2$, $\beta p^{s-1} + 1 \le j \le (\beta + 1)p^{s-1}$\
	be integers. Let $C = \langle (x^n - \xi)^{p^s}(x^n + \xi)^{j} \rangle $.
	Then $d_H(C) = 2(\beta + 2)$.
\end{corollary}

\begin{lemma}\label{Reducible.Lemma.ips.j.tau.K}
	Let $1 \le \tau \le p-1, 1\le k \le s-1,j$\ be integers and
	$C = \langle (x^n - \xi)^{p^s}(x^n + \xi)^{p^s - p^{s-k} + (\tau -1)p^{s-k-1} + 1} \rangle$.
	Then $d_H(C) \ge 2 (\tau + 1)p^k $.
\end{lemma}

\begin{corollary}\label{Reducible.Corollary.ips.j.tau.K}
	Let $1 \le \tau \le p-1, 1\le k \le s-1,j$\ be integers such that
	\be\nn
		p^s - p^{s-k} + (\tau -1)p^{s-k-1} + 1 \le j \le p^s - p^{s-k} + \tau p^{s-k-1}.
	\ee
	Let $C = \langle (x^n - \xi)^{p^s}(x^n + \xi)^j \rangle$.
	Then $d_H(C) = 2 (\tau + 1)p^k$.
\end{corollary}

We summarize our results in the following theorem.

\begin{theorem}\label{Reducible.Theorem.Main}
	Let $p$\ be an odd prime, $a,s,n$\ be arbitrary positive integers and
	$q = p^a$. Let $\lambda, \xi, \psi \in \F_q \setminus \{ 0 \}$\ such that
	$\lambda = \psi^{p^s}$. Suppose that the polynomial $x^{2n} - \psi$\ factors
	into two irreducible polynomials as
	$x^{2n} - \psi = (x^n - \xi)(x^n + \xi)$.
	Then all $\lambda$-cyclic codes, of length $2np^s$, over $\F_q$\
	are of the form $\langle (x^n - \xi)^i(x^n + \xi)^j \rangle \subset
	\F_q[x] / \langle x^{2np^s} - \lambda \rangle$, where
	$0 \le i,j \le p^s$\ are integers.
	Let $C = \langle (x^n - \xi)^i(x^n + \xi)^j \rangle \subset
	\F_q[x] / \langle x^{2np^s} - \lambda \rangle$.
	If $(i,j) = (0,0)$, then $C$\ is the whole space $\F_q^{2np^s}$,
	and if $(i,j) = (p^s,p^s)$, then $C$\ is the zero space $\{\mathbf{0}\}$.
	For the remaining values of $(i,j)$, the minimum Hamming distance of $C$\
	is given in Table \ref{table.Reducible.Hamming.Distance}.
\end{theorem}

\begin{remark}\label{Reducible.Remark.Hamming.Distance.table}
	There are some symmetries in most of the cases, so we made the following simplification 
	in Table \ref{table.Reducible.Hamming.Distance}. 
	For the cases with *, i.e., the cases except 2 and 7,
	we gave the minimum Hamming distance of $C$\ when $i \ge j$. The corresponding case with $j \ge i$
	has the same minimum Hamming distance. For example in 1*, the corresponding case is $i= 0 $\ and 
	$0 \le j \le p^s$, and the minimum Hamming distance is $2$. Similarly in 6*, the corresponding case is
	$\beta p^{s-1} + 1 \le i \le (\beta + 1)p^{s-1}$\ and 
	$p^{s}-p^{s-k}+(\tau-1)p^{s-k-1} + 1 \le j \le p^{s}-p^{s-k}+\tau p^{s-k-1}  $,
	and the minimum Hamming distance is $2(\beta + 2)$.
\end{remark}

\vfill
\pagebreak

\begin{table}\caption{The minimum Hamming distance of all non-trivial constacyclic codes, of the form $\langle (x^n-\xi)^i(x^n+\xi)^j \rangle$, of length $2np^s$\ over $\F_q$. 
The polynomials $x^n-\xi$\ and $x^n + \xi$\ are assumed to be irreducible.
The parameters $1 \le \beta^{'} \le \beta \le p-2$, $ 1 \le \tau^{(2)} < \tau^{(1)} \le p-1$, 
$1 \le \tau, \tau^{(3)}, \tau^{(4)} \le p-1$ , $1 \le k \le s-1$,
	$1 \le k^{''} < k^{'} \le s-1$\ below are integers. For the cases with *, i.e., the cases except 2 and 7, see Remark \ref{Reducible.Remark.Hamming.Distance.table}}
\label{table.Reducible.Hamming.Distance}
\centering
\newcounter{ijCase}
\setcounter{ijCase}{1}
\begin{tabular}{|l|l|l|l|}
        \hline
        \textbf{Case} & \textbf{i} & $\textbf{j}$ & $\textbf{d}_{\textbf{H}}\textbf{(C)}$ \\
        \hline
        \arabic{ijCase}*\addtocounter{ijCase}{1} & $0 < i \le p^s$ & $j=0$ & $2$\\
        \hline
        \arabic{ijCase}\addtocounter{ijCase}{1} & $ 0 \le i \le p^{s-1}$ & $0 \le j \le p^{s-1}$ & $2$\\
        \hline
        \arabic{ijCase}*\addtocounter{ijCase}{1} & $p^{s-1} < i \le 2p^{s-1} $ & $0 < j \le p^{s-1}$ & $3$\\
         \hline \arabic{ijCase}*\addtocounter{ijCase}{1} &
        $ 2p^{s-1} < i \le p^s$ & $0 < j \le p^{s-1} $ & $4$\\
         \hline
        \arabic{ijCase}*\addtocounter{ijCase}{1} & $\beta p^{s-1} + 1 \le i \le (\beta + 1)p^{s-1}$ 
        &  $\beta^{'} p^{s-1} + 1 \le j \le (\beta^{'} + 1)p^{s-1}$ & 
        				$ \begin{array}{l}
                                               \min\{\beta+2,\\ 
                                               2(\beta^{'}+2)\}
                                               \end{array} $\\
        \hline \arabic{ijCase}*\addtocounter{ijCase}{1} &  $
            \begin{array}{l}
                p^{s}-p^{s-k}+(\tau-1)p^{s-k-1} \\ 
                + 1 \le i \le p^{s}-
                p^{s-k}+\tau p^{s-k-1}
            \end{array}$
        & $\beta p^{s-1} + 1 \le j \le (\beta + 1)p^{s-1}$ & $2(\beta + 2)$ \\
        \hline \arabic{ijCase}\addtocounter{ijCase}{1} & 
        $
            \begin{array}{l}
                p^{s}-p^{s-k}+(\tau-1)p^{s-k-1}\\
                + 1 \le i
                \le p^{s}-p^{s-k}+\tau p^{s-k-1}
            \end{array}$
        & $
            \begin{array}{l}
                p^{s}-p^{s-k}+(\tau-1)p^{s-k-1} \\
                + 1 \le j 
                \le p^{s}-p^{s-k}+ \tau p^{s-k-1}
            \end{array}$
        & $(\tau+1)p^{k}$\\
        \hline \arabic{ijCase}*\addtocounter{ijCase}{1} & 
        $
            \begin{array}{l}
                p^{s}-p^{s-k}+(\tau^{(1)}-1)p^{s-k-1}\\
                \begin{array}{ll}
                 	+ 1 \le i 
                	\le & p^{s}-p^{s-k}\\
                	& + \tau^{(1)} p^{s-k-1}
                \end{array}
            \end{array}$
        & $
            \begin{array}{l}
                p^{s}-p^{s-k}+(\tau^{(2)}-1)p^{s-k-1} \\
                \begin{array}{ll}
                	+ 1 \le j 
                	\le & p^{s}-p^{s-k}\\
                	& + \tau^{(2)} p^{s-k-1}
                \end{array}
            \end{array}$
        & $\begin{array}{l}
		\min\{\\  
		2(\tau^{(2)}+1)p^{k},\\
		 (\tau^{(1)} + 1)p^k \}
           \end{array}$\\
        \hline \arabic{ijCase}*\addtocounter{ijCase}{1} & 
        $
            \begin{array}{l}
                p^{s}-p^{s-k^{'}}+(\tau^{(3)}-1)p^{s-k^{'}-1}\\
                \begin{array}{ll}
                 	+ 1 \le i
                	\le & p^{s}-p^{s-k^{'}}\\
                	& + \tau^{(3)} p^{s-k^{'}-1}
                \end{array}
            \end{array}$
        & $
            \begin{array}{l}
                 p^{s}-p^{s-k^{''}}+(\tau^{(4)}-1)p^{s-k^{''}-1}\\
                 \begin{array}{ll}
                	 + 1 \le j
                	\le  & p^{s}-p^{s-k^{''}}\\
                 	& + \tau^{(4)} p^{s-k^{''}-1}
                 \end{array}
            \end{array}$
        & $2(\tau^{(4)}+1)p^{k^{''}} $\\
        \hline \arabic{ijCase}*\addtocounter{ijCase}{1} & 
            $i = p^s$ & $\beta p^{s-1} + 1 \le j \le (\beta + 1) p^{s-1} $ & $2(\beta + 2)$\\
        \hline \arabic{ijCase}*\addtocounter{ijCase}{1} & 
            $i = p^s$ & $\begin{array}{l}
                p^{s}-p^{s-k}+(\tau-1)p^{s-k-1}\\
                \begin{array}{ll}
                	+ 1 \le j 
                	\le & p^{s}-p^{s-k}\\
                	& + \tau p^{s-k-1}
                \end{array}
            \end{array}$ 
        & $2(\tau + 1)p^k$\\
        \hline
\end{tabular}
\end{table}

\section{Examples}
\label{Section.Examples}

We give examples of constacyclic codes of length $np^s$\ and $2np^s$\ that satisfies the conditions
given in Section \ref{Section.Irreducible} and Section \ref{Section.Reducible} respectively.

\begin{example}
First we fix our alphabet as $\F_{16}$. Let $\omega$\ be a generator of the multiplicative group
$\F_{16}\setminus \{ 0 \}$. Having no root in $\F_{16}$, the polynomial $x^3 + \omega^2$\ is
irreducible over $\F_{16}$. Let $s$\ be a positive integer and $\lambda = w^{2\cdot 2^s}$.
Then $\lambda$-cyclic codes, of length $3 \cdot 2^s$\ over $\F_{16}$ correspond to the ideals of the ring
\be\nn
	\sR_1 = \frac{\F_{16}[x]}{\langle x^{3 \cdot 2^s} + \lambda \rangle}.
\ee
So all such $\lambda$-cyclic codes are of the form $\langle (x^3 + \omega^2)^i \rangle$\ for some
$0 \le i \le 2^s$. Let $ C = \langle (x^3 + \omega^2)^i \rangle$.
According to Theorem \ref{Irreducible.Theorem.Main}, the minimum Hamming distance of $C$\ is 
given by
	 $$ 
	\begin{array}{l}
    		\dd d_H(C)
    		 = \left\lbrace
    		\begin{array}{ll}
        		2, &  \mbox{if}\ \ 1 \le i \le 2^{s-1},\\
        		2^{k+1}, & \mbox{if}\ \ 2^s-2^{s-k}+1 \le i \le 2^s-2^{s-k}+ 2^{s-k-1}\\
                	& \mbox{where}\ \ 1 \le k \le s-1.
    		\end{array} \right.
 	\end{array}
 	$$
\end{example}

\begin{example}
	We let the alphabet to be $\F_{13}$. Clearly the polynomials $x^3 - 2$\ and
	$x^3 + 2$\ are irreducible over $\F_{13}$.
	Let $s$\ be a positive integer and $\theta = 4^{13^s}$.
	Then $\theta$-cyclic codes, of length $6\cdot 13^s$, over $\F_{13}$\ correspond to the ideals
	of the ring
	\be\nn
		\sR_2 = \frac{\F_{13}[x]}{\langle x^{6 \cdot 13^s} - \theta \rangle}.
	\ee
	So all such $\theta$-cyclic codes are of the form
	$\langle (x^3 - 2)^i(x^3 + 2)^j \rangle$\ for some integers $0 \le i,j \le 13^s$.
	Using Theorem \ref{Reducible.Theorem.Main}, we determine the minimum Hamming distance
	of $C$\ as in Table \ref{table.Reducible.Hamming.Distance}, where $p$\ is replaced by $13$,
	$n$\ is replaced by $3$\ 
	and $\xi$\ is replaced by $2$.
\end{example}

The following example shows that the main result of \cite{OZOZ_2} is a particular case of Theorem \ref{Reducible.Theorem.Main}.
\begin{example}
	Let $p$\ be an odd prime and $s$\ is a positive integer.
	The cyclic codes, of length $2p^s$, over $\F_q$\ are of the form
	$\langle (x-1)^i(x + 1)^j \rangle$\ where $0 \le i,j \le p^s$.
	Using Theorem \ref{Reducible.Theorem.Main}, we determine the minimum Hamming distance
	of $C$\ as in Table \ref{table.Reducible.Hamming.Distance}, where we consider
	$\xi = 1$, $n = 1$\ and $p$\ to be an odd prime.
\end{example}

\begin{example}
	Let $p$\ be an odd prime. If $p \equiv 1 \mod 4$\ or $a$\ is even, then 
	by Lemma \ref{Irreducible.Lemma.Irreducibility.Criterion}, the polynomial
	$x^2 + 1$\ is reducible over $\F_{p^a}$. Let $\xi$\ and $-\xi$\ be the roots
	$x^2+1$, i.e., $x^2 + 1 = (x - \xi)(x + \xi)$.
	Suppose that $p \equiv 1 \mod 4$\ or $a$\ is even. Then the negacyclic codes, 
	of length $2p^s$, over $\F_{p^a}$\ are of the form
	$\langle (x- \xi)^i(x + \xi)^j \rangle$\ where $0 \le i,j \le p^s$. 
	Let $C = \langle (x- \xi)^i(x + \xi)^j \rangle$. Then the minimum
	Hamming distance of $C$\ is given in Table \ref{table.Reducible.Hamming.Distance}
	where we consider $n = 1$.
\end{example}



\begin{thebibliography}{99} \label{section.bib}

\bibitem{BRLKMP} E. Berlekamp, ``Algebraic Coding Theory", Mc-Graw Hill, 1968.

\bibitem{CMSS1} G. Castagnoli, J. L. Massey, P. A. Schoeller, N. von Seemann, ``On repeated-root cyclic codes", \emph{IEEE Trans. 	Inform. Theory},
	vol. 37, pp. 337-342, 1991.

\bibitem{CMSS2} D. J. Costello, J. Justesen, J. L. Massey ``Polynomial weights and code constructions", \emph{IEEE Trans. Inform. Theory},
vol. 19, pp. 101-110, 1973.

\bibitem{D2007} H. Q. Dihn, ``Complete distances of all negacyclic codes of length $2^s$\ over $\Z_{2^a}$", \emph{IEEE Trans. Inform. Theory}, vol. 53, pp. 147-161, 2007.

\bibitem{D2008} H. Q. Dihn, ``On the linear ordering of some classes of negacyclic and cyclic codes and their distance distributions",
	\emph{Finite Fields Appl.}, vol. 14,  pp. 22-40, 2008.

\bibitem{OZOZ_1} H. \"Ozadam and F. \"Ozbudak, ``A note on negacyclic and cyclic codes of length $p^s$
      over a finite field of characteristic $p$ ", submitted, 2009.

\bibitem{OZOZ_2} H. \"Ozadam and F. \"Ozbudak, ``The minimum Hamming distance of cyclic codes of length $2p^s$ ", proceedings of AAECC-18, Springer LNCS, vol 5527, pp. 92-100, 2009.




\end{thebibliography}
\end{document}